\documentclass[11pt,a4paper,reqno]{amsart}

\usepackage[english]{babel}
\usepackage[latin1]{inputenc}
\usepackage[hidelinks]{hyperref}
\usepackage{amsmath,
latexsym,
amssymb,
amscd,
amsgen,
amstext,
amsbsy,
amsopn,
math rsfs,
bm,
bbm,
amsthm,
epsfig,
graphicx,
graphics,
xspace,
amsxtra,
xcolor,
dsfont,
enumerate,
mathabx}

\usepackage[deletedmarkup=xout,
]{changes}
\definechangesauthor[name={Michele}, color={red}]{michele}
\definechangesauthor[name={Lorenzo}, color={blue}]{lorenzo}
\definechangesauthor[name={Raffaele}, color={green}]{raffaele}

\usepackage{geometry}
\makeatletter
\renewcommand*{\@textcolor}[3]{%
  \protect\leavevmode
  \begingroup
    \color#1{#2}#3%
  \endgroup
}
\makeatother

\newtheorem{teo}{Theorem}[section]
\newtheorem{lem}{Lemma}[section]
\newtheorem{pro}{Proposition}[section]

\theoremstyle{remark}
\newtheorem{rem}{Remark}[section]

\newcommand{\qv}{\mathbf{q}}
\newcommand{\zev}{\mathbf{0}}
\newcommand{\xv}{\mathbf{x}}
\newcommand{\yv}{\mathbf{y}}
\newcommand{\pv}{\mathbf{p}}

\newcommand{\R}{\mathbb{R}}
\newcommand{\C}{\mathbb{C}}
\newcommand{\dt}{\,\mathrm{d}t}
\newcommand{\ds}{\,\mathrm{d}s}
\newcommand{\dk}{\,\mathrm{d}k}
\newcommand{\dtau}{\,\mathrm{d}\tau}

\newcommand{\dsigma}{\,\mathrm{d}\sigma}
\newcommand{\dmu}{\,\mathrm{d}\mu}
\newcommand{\dpi}{\,\mathrm{d}p}
\newcommand{\dpv}{\,\mathrm{d}\pv}
\newcommand{\dxv}{\,\mathrm{d}\xv}
\newcommand{\drho}{\,\mathrm{d}\varrho}
\newcommand{\I}{\mathcal{I}}
\newcommand{\J}{\mathcal{J}}

\newcommand{\NN}{\mathcal{N}}
\newcommand{\OO}{\mathcal{O}}

\newcommand{\F}{\mathcal{F}}

\newcommand{\KK}{\mathcal{K}}

\newcommand{\one}{\mathds{1}}

\newcommand{\h}{\widetilde{h}}
\newcommand{\T}{\widetilde{T}}

\newcommand{\dom}{\mathscr{D}}
\newcommand{\TF}{\mathscr{F}}

\newcommand{\ep}{\varepsilon}

\newcommand{\clog}{C_{\mathrm{log},\beta}}

\newcommand{\fe}{f_{\mathrm{e}}}
\renewcommand{\Re}{\mathrm{Re}}

\numberwithin{equation}{section}
\newcommand{\bdm}{\begin{displaymath}}
\newcommand{\edm}{\end{displaymath}}
\newcommand{\bdn}{\begin{eqnarray}}
\newcommand{\edn}{\end{eqnarray}}
\newcommand{\bay}{\begin{array}{c}}
\newcommand{\eay}{\end{array}}
\newcommand{\ben}{\begin{enumerate}}
\newcommand{\een}{\end{enumerate}}
\newcommand{\beq}{\begin{equation}}
\newcommand{\eeq}{\end{equation}}
\newcommand{\beqn}{\begin{eqnarray}}
\newcommand{\eeqn}{\end{eqnarray}}
\newcommand{\bml}[1]{\begin{multline} #1 \end{multline}}
\newcommand{\bmln}[1]{\begin{multline*} #1 \end{multline*}}

\renewcommand{\leq}{\leqslant}
\renewcommand{\geq}{\geqslant}

\newcommand{\disp}{\displaystyle}
\newcommand{\tx}{\textstyle}
\newcommand{\lf}{\left}
\newcommand{\ri}{\right}

\newcommand{\braket}[2]{\lf\langle #1|#2 \ri\rangle}

\newcommand{\mean}[1]{\lf\langle #1 \ri\rangle}

\newcommand{\form}{\F_{\alv(t)}}
\newcommand{\ham}{H_{\alv(t)}}

\newcommand{\alv}{\bm{\alpha}}
\newcommand{\phila}{\phi_{\lambda}}

\newcommand{\la}{\lambda}
\newcommand{\diff}{\mathrm{d}}
\newcommand{\tv}{\mathbf{t}}
\newcommand{\BB}{\mathcal{B}}
\newcommand{\A}{\mathcal{A}}
\newcommand{\G}{\mathcal{G}}
\newcommand{\eps}{\varepsilon}

\begin{document}
 
 \title[Well-posedness of 2D NLSE with concentrated nonlinearity]{Well-posedness of the two-dimensional nonlinear Schr\"odinger equation with concentrated nonlinearity}

\author[R. Carlone]{Raffaele Carlone}
\address{Universit\`{a} ``Federico II'' di Napoli, Dipartimento di Matematica e Applicazioni ``R. Caccioppoli'', MSA, via Cinthia, I-80126, Napoli, Italy.}
\email{raffaele.carlone@unina.it}
\author[M. Correggi]{Michele Correggi}
\address{``Sapienza'' Universit\`{a} di Roma , Dipartimento di Matematica, P.le Aldo Moro, 5, 00185, Roma, Italy.}
\email{michele.correggi@gmail.com}
\urladdr{http://www1.mat.uniroma1.it/people/correggi/}
\author[L. Tentarelli]{Lorenzo Tentarelli}
\address{``Sapienza'' Universit\`{a} di Roma , Dipartimento di Matematica, P.le Aldo Moro, 5, 00185, Roma, Italy.}
\email{tentarelli@mat.uniroma1.it}

\date{\today}

\begin{abstract} 
 We consider a two-dimensional nonlinear Schr\"odinger equation with concentrated nonlinearity. In both the focusing and defocusing  case we prove local well-posedness, i.e., existence and uniqueness of the solution for short times, as well as energy and mass conservation. In addition, we prove that this implies global existence in the defocusing case, irrespective of the power of the nonlinearity, while in the  focusing  case blowing-up solutions may arise.
\end{abstract}

\maketitle

\tableofcontents


\section{Introduction and Main Results}

The nonlinear Schr\"{o}dinger (NLS) equation plays a relevant role in several sectors of physics, where it  appears very often as an effective evolution equation describing the behavior of a microscopic system on a macroscopic or mesoscopic scale. A typical example is provided by the time evolution of Bose-Einstein condensates, which is known to be well approximated by a NLS-type equation going under the name of Gross-Pitaevskii equation \cite{DGPS}. There are however other examples in which the physical meaning of the NLS equation is totally different, as, e.g., the propagation of light in nonlinear optics, the behavior of water or plasma waves, the signal transmission through neurons (FitzHugh-Nagumo model), etc. (see, e.g., \cite{Ma} and references therein).

Thanks to its physical relevance, the NLS equation has attracted a lot of interest within the mathematical community as well, and several monographs are devoted to its detailed study (see, e.g., \cite{C}). Here we focus on the simple but nontrivial case of a nonlinearity affecting the evolution only at finitely many points, i.e., a NLS equation with {\it concentrated nonlinearity}. Roughly speaking the model we want to investigate is described by the two-dimensional formal equation
\beq
	\label{eq:formalNLS}
	i \partial_t \psi_t = \bigg(-\Delta + \sum_{j  =1}^N \mu_j \delta(\xv - \yv_j) \bigg) \psi_t,
\eeq
where any coupling parameter $ \mu_j = \mu_j(\psi_t(\yv_j)) $ depends itself on the value of the function $ \psi_t $ at $ \yv_j $ (see below). 

Such a model has been used in physics to describe very different phenomena, mostly related to solid state physics: the charge accumulation in semiconductor interfaces or heterostructures can be modelled indeed by nonlinear effects concentrated in a small spatial region \cite{BKB,JLPC,JLPS,MA,N}. The idea is that the nonlinear term takes into account the many-body interaction effects on the scattering of an electron through a barrier or by an impurity in the medium \cite{MB}. In nonlinear optics similar models arise in the description of the nonlinear propagation in a Kerr-type medium in presence of localized defects \cite{SKB,SKBRC,Y}, but several other applications are suggested in acoustic, conventional and high-$T_c$ superconductivity, light propagation in photonic crystals etc. (see \cite{SKBRC} and references therein). More recently the nonlinear propagation in presence of a concentrated defect has been suggested as a dynamic model for the evolution of Bose-Einstein condensates in optical lattices, where the isolated defect is generated by a focused laser beam \cite{DM,LKMF}. 

From the mathematical point of view, the expression between brackets in the above formula is purely formal, at least in two or more dimensions, and, in order to give it a rigorous meaning, one can follow different paths, as, e.g., classifying the self-adjoint extensions of suitable symmetric operators \cite{Al} or investigating the properties of the associated quadratic forms \cite{DFT}. The reason why such models (a.k.a. {\it solvable models}), involving zero-range or point interactions, have attracted so much interest in the past is that the time evolution described by \eqref{eq:formalNLS} can be simplified and in fact reduced to an ODE-type evolution of finitely many complex numbers named {\it charges} (see below), which are proportional to the values of $ \psi_t $ at the singular points. This was first observed in the corresponding time-dependent linear models \cite{DFT2,SY} (see also \cite{CD,CDFM,CCF,CCNP,NP} for similar results) and later used also in the nonlinear framework.

Analogous 1 and 3D models  have indeed already been studied in details in the literature \cite{A,AT,ADFT1,ADFT2}: it has been proven that the weak Cauchy problem associated to \eqref{eq:formalNLS} in 1 or 3D (or rather to its rigorous analogue) admits a unique solution in the proper energy space for short times and that, under additional assumptions of the parameter (e.g., in the defocusing case), such a solution is in fact global in time (thanks to the mass and energy conservation). Further results about the possible emergence of blow-up solutions have also been established; so that the 1 and 3D models are basically completely understood. On the opposite, no results about the well-posedness (neither local nor global) of the 2D equation are available so far, mostly due to hard technical difficulties emerging in 2D (see the discussion at the end of next Sect.). It is also worth mentioning that the 1D and 3D analogues of the model above have been rigorously derived in \cite{CFNT1,CFNT} from the ordinary NLS equation in a suitable scaling limit of nonlinearity concentration. 

\subsection{The model}
\label{model}
We specify now more precisely the model we want to investigate. We are interested in discussing a specific form of 2D NLS equation with concentrated nonlinearities at finitely many points $ \yv_1, \ldots, \yv_N \in \R^2 $, with $ \yv_i \neq \yv_j $ for $ i \neq j $. The precise definition of the model is similar to the 3D one, with some small but relevant modifications mostly due to the peculiar behavior of the 2D Green function (see below). 

We start by recalling the properties of the linear version of the evolution problem \eqref{eq:formalNLS}, which has been studied in \cite{CCF}: the idea (see \cite[Theorem 5.3]{Al} for further details) is to reformulate \eqref{eq:formalNLS} as the Schr\"{o}dinger equation $	i \partial_t \psi_t = \ham \psi_t $
associated to a time-dependent Schr\"{o}dinger operator $ \ham $ on $ L^2(\R^2) $, defined as
\beq
	\label{eq:lin op}
	\lf(\ham + \la \ri) \psi  = \lf(- \Delta + \la \ri) \phila,	
\eeq
with domain
\bml{
	\label{eq:lin domain}
	\dom\lf(\ham\ri) = \bigg\{ \psi \in L^2(\R^2) \: \big| \: \psi = \phila + \frac{1}{2\pi} \sum_{j=1}^N q_j(t) K_0 \lf(\sqrt{\la} |\xv - \yv_j|\ri), \phila \in H^2(\R^2), 	\\
	\disp\lim_{\xv \to \yv_j} \phila(\xv) = \left(\alpha_j(t)+\tx\frac{1}{2\pi} \log \frac{\sqrt{\la}}{2} + \frac{\gamma}{2\pi}\right)q_j(t)-\frac{1}{2\pi}\sum_{k\neq j} K_0 \lf(\sqrt{\la} |\yv_j - \yv_k|\ri) q_k(t) \bigg\},
}
where $ \la > 0 $, $K_0(\sqrt{\la} |\xv|)$ denotes the inverse Fourier transform of $ (|\pv|^2 + \lambda)^{-1} $, i.e., the modified Bessel function of second kind of order $ 0 $ (a.k.a. Macdonald function \cite[Sect. 9.6]{AS}), $\gamma$ is the Euler constant and the function $ \alv(t) = (\alpha_1(t), \ldots, \alpha_N(t)) $ is assumed to be of class $C^1$.

Wave functions in the operator domain are thus decomposable into a regular part $ \phila $, belonging to the domain of the free Laplacian, plus a more singular term proportional to the Green function of $ - \Delta+\lambda $, which shows logarithmic singularities at the points $ \yv_1, \ldots, \yv_N $ (see \cite{CCF}). The interaction is replaced with a boundary condition linking the values of the regular part $ \phila $ at points $ \yv_1, \ldots, \yv_N $ to the coefficients of the singular one.

  	\begin{rem}[Domain decomposition]
  		\mbox{}	\\
  		In the definition of the domain \eqref{eq:lin domain} a first difference with the 3D case emerges: the operator domain $ \dom(\ham) $ is obviously independent of the parameter $ \lambda $, but, while in 3D one is allowed to take $ \lambda = 0 $ (with some little care about the large $ |\xv| $ decay), the same is not possible in 2D. Due to its infrared singularity, the 2D Green function actually diverges when $ \lambda \to 0 $ and therefore such a choice is forbidden.
	\end{rem}

The Cauchy problem for the linear evolution equation, i.e.,
\beq
	\label{eq:cauchy_lin_op}
	\begin{cases}
		 i \partial_t \psi_t =  \ham \psi_t,	\\
		\psi_{t = 0}  =  \psi_0,	
	\end{cases}
\eeq
with $ \psi_0 \in \dom(H_{\alv(0)}) $, was studied in \cite{CCF}, where it was proven that $ \ham $ generates a two-parameter unitary group $ U(t,s) $ and therefore, if $ \phi \in \dom(H_{\alv(0)}) $, then also $ \psi_t \in \dom(\ham) $ for any time $  t \in \R $.

Equivalently {(see \cite{DFT})} one can consider the quadratic form $ \form $ associated to the operator $ \ham $,
\begin{align}
\label{eq:from_pre}
	\form[\psi] := & \, \int_{\R^2} \dxv \lf\{ \lf| \nabla \phila \ri|^2 + \la \lf| \phila \ri|^2 - \la |\psi|^2 \ri\} + \sum_{j = 1}^N \lf( \alpha_j(t) + \tx\frac{1}{2\pi} \log \frac{\sqrt{\la}}{2} + \frac{\gamma}{2\pi} \ri) |q_j|^2 	\nonumber\\
	& \, - \frac{1}{2\pi} \sum_{j=1}^N\sum_{k\neq j} q^*_{j} q_{k} K_{0}(\sqrt{\lambda}\lf| \yv_j - \yv_k \ri|)
\end{align}
with time-independent domain
\[
	\dom[\F] 
	:= \bigg\{ \psi \in L^2(\R^2) \: \big| \: \psi = \phila + \frac{1}{2\pi} \sum_{j = 1}^N q_j K_0\lf(\sqrt{\la} |\xv - \yv_j|\ri), \phila \in H^1(\R^2), q_j \in \C \bigg\},
\] 
and the weaker version of the Cauchy problem \eqref{eq:cauchy_lin_op}:
\beq
	\label{eq:cauchy_lin_form}
	\begin{cases}
		 i \frac{d}{dt} \braket{\chi}{\psi_t} =  \form[\chi, \psi_t],\qquad\forall\chi \in \dom[\F],	\\
		 \psi_{t=0}  =  \psi_0,	
	\end{cases}
\eeq
where the initial datum $ \psi_0 $ also belongs to the form domain $ \dom[\F] $, $ \braket{\: \cdot \:}{\: \cdot \:} $ stands for the scalar product in $ L^2(\R^2) $ and $ \form[ \: \cdot \:, \: \cdot \:] $ is the sesquilinear form associated to $ \form $ defined, e.g., by polarization. The well-posedness of the above Cauchy problem is also proven in \cite{CCF}. Note that, unlike the operator domain, functions in the form domain $\dom[\F]$ do not have to satisfy any boundary condition.

A solution to both linear problems \eqref{eq:cauchy_lin_op} and \eqref{eq:cauchy_lin_form} (see \cite[Sect. 2.2]{CCF}) is provided by the following ansatz
\begin{equation}
 \label{eq:ansatz_pre}
	\psi_t (\xv) : = \lf( U_0(t) \psi_0 \ri) (\xv) + \disp\frac{i}{2\pi} \sum_{j=1}^N \int_0^t \diff \tau \: U_0\lf(t - \tau; |\xv - \yv_j|\ri) \: q_j(\tau),
\end{equation}
 where $U_0(t)=e^{ i \Delta t}$ denotes the free propagator, whose integral kernel is given by
\[
 U_0(t; |\xv|) = \frac{e^{-\frac{|\xv|^{2}}{4 i t}}}{2 i t},\qquad t \in \R,\:\xv \in \R^2,
\]
and the function $ \qv(t) = (q_1(t), \ldots, q_N(t)) $ is the solution of a Volterra-type equation of the form
\[
	q_j(t) + 4 \pi \int_0^t \diff \tau \: \I(t - \tau) \: \alpha_j(\tau) q_j(\tau) + \sum_{k = 1}^N \disp\int_0^t \diff \tau \: \KK_{jk}( t-\tau ) \: q_k(\tau) = f_j(t)
\]
(see below for more details).

The nonlinear model we plan to investigate in this article is the analogue of \eqref{eq:cauchy_lin_op} (resp. \eqref{eq:cauchy_lin_form}), where the parameters $ \alv(t) $ depend themselves on the values of the charge $\qv(t)$, i.e.,
\beq
	\label{eq:alpha}
	\boxed{
	\alpha_j(t) = \beta_j \lf| q_j(t) \ri|^{2\sigma_j},	\qquad		\beta_j \in \R, \sigma_j \in \R^+. 
}
\eeq
Hence for any wave function in the nonlinear operator domain, the above nonlinearity can be translated into $ N $ nonlinear boundary conditions, i.e.,
\beq
	\label{eq:boundary_condition}
	\disp\lim_{\xv \to \yv_j} \phila(\xv) = \left(\beta_j \lf| q_j(t) \ri|^{2\sigma_j}+\tx\frac{1}{2\pi} \log \frac{\sqrt{\la}}{2} + \frac{\gamma}{2\pi}\right)q_j(t)-\frac{1}{2\pi}\sum_{k\neq j} K_0(\sqrt{\la} |\yv_j - \yv_k|)q_k(t).
\eeq

Consequently, the goal is to prove that a weak solution to the Cauchy problem \eqref{eq:cauchy_lin_op} (i.e., a solution to \eqref{eq:cauchy_lin_form}) with the {\it nonlinear} condition \eqref{eq:alpha} is provided by the very same ansatz as in \eqref{eq:ansatz_pre}, i.e.,
\bdm
	\psi_t (\xv) = \lf( U_0(t) \psi_0 \ri) (\xv) + \disp\frac{i}{2\pi} \sum_{j=1}^N \int_0^t \diff \tau \: U_0\lf(t - \tau; |\xv - \yv_j|\ri) \: q_j(\tau),
\edm
where $ \qv(t) $ is now the solution of the Volterra-type {\it nonlinear} equation
\begin{equation}
 \label{eq:charge_eq_pre}
 	\begin{array}{l}
  	\displaystyle q_j(t)+4\pi\beta_j\int_0^t\dtau\:\I(t-\tau)|q_j(\tau)|^{2\sigma_j}q_j(\tau) \\
  \hspace{2cm} \displaystyle + \sum_{k=1}^N\int_0^t\dtau\:\KK_{jk}(t-\tau)q_k(\tau)= \underbrace{4\pi\int_0^t\dtau\:\I(t-\tau)(U_0(\tau)\psi_0)(\yv_j)}_{{:=f_j(t)}},
 \end{array}
\end{equation}
with $\I$ denoting the Volterra function of order $-1$
\begin{equation}
 \label{eq:i}
 \I(t) : = \int_{0}^{\infty}\dtau \: \frac{t^{\tau - 1}}{\Gamma(\tau)},
\end{equation}
and where $\KK_{jk}$, $ j, k = 1, \ldots, N $, is defined by
\begin{equation}
 \label{eq:KK}
 \KK_{jk}(t) : = 
 \begin{cases}
  \displaystyle -2 \left( \log 2 - \gamma +\tfrac{i\pi}{4} \right) \I(t), & \mbox{if } j = k, \\[.3cm]
  \displaystyle 2i\int_0^t\dtau\:\I(t-\tau)U_0(\tau;|\yv_j-\yv_k|), & \mbox{if } j \neq k.
 \end{cases}
\end{equation}
Notice that the choice of the initial time $t=0$ is completely arbitrary: everything we prove in this paper holds as well if the initial time $ t = 0 $ is replaced with any $ s \geq 0 $.

For the sake of completeness we also formulate the weak counterpart of the evolution problem \eqref{eq:ansatz_pre} and \eqref{eq:charge_eq_pre}, which reads as follows: let the initial datum $ \psi_0 $ belong to the form domain $ \dom[\F] $, then $\psi_t\in \dom[\F]$ and 
\[
	\begin{cases}
		 i \frac{d}{dt} \braket{\chi}{\psi_t} =  \lf. \F_{\alv(t)}\big[\chi,\psi_t\big] \ri|_{\lf\{ \alpha_j(t) = \beta_j \lf| q_j(t) \ri|^{2\sigma_j}, \: j = 1,\ldots,N \ri\}},	\\
		 \psi_{t = 0}  =  \psi_0,	
	\end{cases}
\]
for any $ \chi \in \dom[\F] $.

	The form of the Volterra equation \eqref{eq:charge_eq_pre} makes apparent a major difference with the 1 and 3D cases, which is also one of the main reasons why the 2D one called for a more refined analysis: the integral operator with kernel $ \I(t-\tau) $ defined in \eqref{eq:i} is a characteristic feature of the 2D problem and poses hard technical issues (see, e.g., \cite{CF}). In 3D (in 1D the equation is even simpler) the role of  $ I $ is played by the Abel-$1/2$ integral operator with kernel $ 1/\sqrt{t- \tau} $, which enjoys a lot of useful regularizing properties, making the investigation of \eqref{eq:charge_eq_pre} much easier. In that case, by taking smooth enough initial data, the regularity easily propagates to $ \qv(t) $, so that the ansatz \eqref{eq:ansatz_pre} belongs to the operator domain and therefore it provides strong solution to the Cauchy problem. The extension to rougher initial data is then obtained by density. In 2D already the first step, i.e., the regularity of $ \qv(t) $ for smooth initial data, is challenging and the whole proof strategy has to be dramatically changed (see Sect. \ref{sec:proofs}).
	
	Moreover, the lack of regularity of $ \qv(t) $ prevents the use of any density argument, which is precisely the route followed in 1 and 3D: indeed, it is impossible in 2D to restrict the set of initial data, prove the well-posedness and then extend the result to all initial data by density. On the opposite, our strategy relies on a contraction argument, which does not allow to propagate any additional regularity from the initial datum to $ \qv(t) $ (and then to $ \psi_t $). In addition, the appropriate contraction space is $ H^{1/2}(0,T) $, which is known to have a sort of pathological behavior, i.e., failure of the Hardy inequality, absence of natural extensions to $ H^{1/2}(\R) $ etc. (see below), and makes the technical side of the proof really tricky.
	 

\subsection{Main results}
\label{sec:mainres}

{Although, as we pointed out in Section \ref{model}, the model makes sense for an arbitrary number $N$ of nonlinear point interactions, in the following of the paper we will only deal with the case $N=1$. The reason of such a restriction is purely technical (see also the end of Section 2.2). Indeed, with some tricky calculations, it is possible to check that, due to the asymptotic diverging behavior of $\I$ (see next \eqref{eq:Iasympt}) and $U_0(\,\cdot\,;|\yv_j-\yv_k|)$ (for $j\neq k$) near the origin, the off-diagonal kernels in \eqref{eq:KK} are very singular, i.e.,
\[
 	2i\int_0^t\dtau\:\I(t-\tau)U_0(\tau;|\yv_j-\yv_k|)\notin L^1(0,T).
\]
This represents a main issue for the study of the integral equation \eqref{eq:charge_eq_pre}, since all the ``contractive theory'' developed in Section \ref{sec:preliminaries} strongly relies on the integrability of the kernel.}

{It is also worth stressing that in the 1 and the 3D cases such a problem is not present, since if we replace $\I$ with the $\frac{1}{2}$-Abel kernel and $U_0$ with the one- or three-dimensional kernel of the free propagator, then the off-diagonal terms of the charge equation present the same qualitative behavior of the diagonal part.}

\medskip
{Consequently, our discussion only focuses on the case of a single nonlinear interaction placed at $\yv\in\R^2$, with coefficient $\beta_0\in\R$ and nonlinear power $\sigma\in\R^+$ (i.e., $\alpha(t)=\beta_0|q(t)|^{2\sigma}q(t)$); namely we study the properties of the function}
\begin{equation}
 	\label{eq:ansatz}
 	\psi_t (\xv) : = \lf( U_0(t) \psi_0 \ri) (\xv) + \disp\frac{i}{2\pi} \int_0^t \diff \tau \: U_0\lf(t - \tau; |\xv - \yv|\ri) \: q(\tau),
\end{equation}
where $q$ satisfies
\begin{equation}
 \label{eq:charge_eq}
 \boxed{
 \begin{array}{l}
 	\displaystyle q(t)+4\pi\beta_0\int_0^t\dtau\:\I(t-\tau)|q(\tau)|^{2\sigma}q(\tau) \\
	\hspace{1.5cm} \displaystyle -2 \left( \log 2 - \gamma +\tfrac{i\pi}{4}\right)\int_0^t\dtau\:\I(t-\tau)q(\tau)= 4\pi\int_0^t\dtau\:\I(t-\tau)(U_0(\tau)\psi_0)(\yv).
 \end{array}
 }
\end{equation}
More precisely, recalling that in the case of a single point interaction the quadratic form defined by \eqref{eq:from_pre} reads
\[
	\mathcal{F}_{\alpha(t)}[\psi] := \int_{\R^2} \dxv \lf\{ \lf| \nabla \phila \ri|^2 + \la \lf| \phila \ri|^2 - \la |\psi|^2 \ri\} + \lf( \alpha(t) + \tx\frac{1}{2\pi} \log \frac{\sqrt{\la}}{2} + \frac{\gamma}{2\pi} \ri) |q|^2,
\]
with domain
\begin{equation}
\label{eq:form_domain}
\dom[\F]:= \bigg\{ \psi \in L^2(\R^2) \: \big| \: \psi = \phila + \frac{1}{2\pi}qK_0\lf(\sqrt{\la} |\xv - \yv|\ri), \phila \in H^1(\R^2), q \in \C \bigg\},
\end{equation} 
we will show that $\psi_t$ is the unique solution of the weak Cauchy problem
\begin{equation}
\label{eq:cauchyweak}
	\boxed{\begin{cases}
		i \frac{d}{dt} \braket{\chi}{\psi_t} =  \lf. \F_{\alpha(t)}\big[\chi,\psi_t\big] \ri|_{\lf\{ \alpha(t) = \beta_0 \lf| q(t) \ri|^{2\sigma}\ri\}},	\\
		\psi_{t = 0}  =  \psi_0,
	\end{cases}}
\end{equation}
{for any $ \chi=\chi_\lambda+\frac{1}{2\pi} q_{\chi} K_0(\sqrt{\lambda}|\cdot-\yv|) \in \dom[\F] $, where
\bmln{
  \lf.\F_{\alpha(t)}\big[\chi,\psi_t\big] \ri|_{\lf\{ \alpha(t) = \beta_0 \lf| q(t) \ri|^{2\sigma}\ri\}} := \, \int_{\R^2}\dxv\:\left\{\nabla\chi_\lambda^*\cdot\nabla\phi_{\lambda,t}+\lambda\chi_\lambda^*\phi_{\lambda,t}-\lambda\chi^*\psi_t\right\}\\
                                                                                                                                   \, +\bigg(\beta_0|q(t)|^{2\sigma}+\frac{1}{2\pi}\log\frac{\sqrt{\lambda}}{2}+\frac{\gamma}{2\pi}\bigg)q_{\chi}^* q(t).}}

\subsubsection{Local well-posedness and conservation laws}

The first result we prove concerning the evolution problem described above is a local well-posedness for initial data in a suitable subset of the form domain that we define as follows (we set $ p = |\pv| $ for short)
\begin{equation}
 	\label{eq:in_assumption}
 	\dom := \lf\{\psi \in \dom[\F] \: \Big| \: \lf(1 + p^{\epsilon}\ri) \: \widehat{\phi_\lambda}(\pv) \in L^1(\R^2), \mbox{ for some } \epsilon > 0 \ri\},
\end{equation}
where $\widehat{\phi_\lambda}$ stands for the Fourier transform (see \eqref{eq:fourier}) of the regular part $\phi_\lambda$.

	\begin{teo}[Local well-posedness]
		\label{teo:local}
		\mbox{}	\\
		Let $ \psi_0 \in \dom $ and $ \sigma \geq \frac{1}{2} $. Then, there exists $ T > 0 $ such that there is a unique solution to \eqref{eq:cauchyweak} belonging to $ \dom[\F] $ for any $ t \leq T $ and it is given by \eqref{eq:ansatz}, with $ q(t) $ the unique solution to \eqref{eq:charge_eq}.
	\end{teo}
	
	\begin{rem}[Charge $ q(t) $]
		\mbox{}	\\
		The above Theorem contains in fact two results: the most important one is the local well-posedness of the weak Cauchy problem \eqref{eq:cauchyweak}, but that result actually follows from the properties of the solution to the Volterra-type equation \eqref{eq:charge_eq}. In fact, once established the existence and uniqueness of $ q(t) $ in $ C[0,T] \cap H^{1/2}(0,T) $ (see Propositions \ref{pro:charge continuous}, \ref{pro:charge} and \ref{pro:extension q}), one can prove that such a regularity transfers to the wave function defined by \eqref{eq:ansatz} and then, thanks to the special form of \eqref{eq:ansatz}, that $\psi_t$ solves \eqref{eq:cauchyweak}. It has to be stressed that the regularity of $ q $ is, in fact, borderline to make this argument work and a very fine analysis has to be performed.
	\end{rem}
	
	\begin{rem}[Uniqueness of $ \psi_t $]
		\mbox{}	\\
		One could think that the ansatz \eqref{eq:ansatz} might not be the unique solution of the weak problem \eqref{eq:cauchyweak}. However, it is easy to see that this is not the case and $ \psi_t $ is in fact the unique solution of \eqref{eq:cauchyweak}. Suppose that, for a given initial datum $ \psi_0 \in \dom $, there was another solution $ \widetilde\psi_t $. Then, by definition, it should decompose as
		$$
			\widetilde\psi_t = \widetilde\phi_{\la,t} + \frac{1}{2\pi} \widetilde{q}(t) K_0 \lf(\sqrt{\la} |\xv - \yv|\ri),
		$$
		for some bounded charge $ \widetilde{q}(t) $ different from $ q(t) $. However, one could as well decompose $ \widetilde\psi_t $ as (see, e.g., Sect. \ref{sec:charge})
		$$
			\widetilde\psi_t(\xv)  = \chi_{\la,t}(\xv) + \disp\frac{i}{2\pi}\int_0^t \diff \tau \: U_0\lf(t - \tau; |\xv - \yv|\ri) \widetilde{q}(\tau),
		$$
		for some function $ \chi_{\la,t} $. Now, it is not difficult to see (Sect. \ref{sec:charge}) that this function can solve \eqref{eq:cauchyweak} if and only if $ \chi_{\la,t} = U_0(t) \psi_0 $ and $ \widetilde q $ solves the charge equation \eqref{eq:charge_eq}. Uniqueness of the solution of \eqref{eq:charge_eq} implies then the result. In fact, in Sect. \ref{sec:charge} the previous argument is carried out in the case of strong solutions, following the original proof of \cite{A} for the linear problem (and with the extra assumption $q(0)=0$). However, it is possible to prove that it can be adapted to weak solutions.
	\end{rem}
	 
	 \begin{rem}[Condition on the nonlinearity]
	 	\mbox{}	\\
	 	Although not so relevant for most physical applications, it is worth discussing briefly the role of the condition $ \sigma \geq \frac{1}{2} $. There is no analogue of such a condition in the proof of local well-posedness for the 1D and 3D models. We believe it is only a technical assumption needed in a single step of the proof. More precisely it is due to the different strategy we have to follow in the first part of the proof, i.e., the contraction argument used in the analysis of the charge equation, which requires to assume $ \sigma \geq \frac{1}{2} $ (see Lemma \ref{lem-lip} and Remark \ref{eq:ass_on_sigma}). Obviously, the case $ \sigma = 0 $ is also included, corresponding to the linear evolution problem studied in \cite{CCF}. 
	 \end{rem}
	 
	 \begin{rem}[Condition on the initial state]
	 	\mbox{}	\\
	 	We point out that the assumption on the initial state $ \psi_0 \in \dom \subsetneq \dom[\F] $ is more restrictive then one would expect, since not only {$ \phi_{\lambda,0} \in H^1(\R^2)$}, but also the Fourier transform must be in $ L^1(\R^2) $. This, for instance, ensures that the time-evolution of the regular part $ U_0(t) \phi_{\lambda,0} $ is a continuous function, in order to be able to evaluate it at the singular point $ \yv$. The condition is deeply related to the lack of regularizing properties of the operator $ I $ and in this respect the choice of $ \dom $ is the most reasonable. {On the other hand, the further requirement $p^\ep\widehat{\phi_{\lambda,0}}\in L^1(\R^2)$ plays a role only in Lemma \ref{lemma:forzcont} (which is however mandatory for the proof of next Theorem \ref{teo:conservation}).} No analogue of these conditions is however present in the 1 and 3D cases and it might as well be that such extra assumptions are not needed for a weak solution. 
	 \end{rem}
	 
In addition, we can claim a conservation result, that also plays a crucial role in the proof of the global existence of the solutions mentioned above:

	\begin{teo}[Conservation laws]
		\label{teo:conservation}
		\mbox{}	\\	
		Let $ \psi_0 \in \dom $, $\psi_t$ be the wave function defined by \eqref{eq:ansatz} and \eqref{eq:charge_eq} and $ T > 0  $ the existence time provided by Theorem \ref{teo:local}. Then, the mass $ M(t)  = \lf\| \psi_t \ri\|_2 $ and the energy
		\begin{equation}
			\label{eq:energy}
			E(t) = \|\phi_{1,t}\|_{H^1(\R^2)}^2 +\left( \frac{\beta_0}{\sigma+1}|q(t)|^{2\sigma}+ \frac{\gamma-\log 2}{2\pi}\right)	|q(t)|^2
		\end{equation}
		are conserved for every $t\in[0,T]$.
	\end{teo}
	
	\begin{rem}[Dependence on $ T $]
		\mbox{}	\\
	 We stress that, as it is, the conservation of the mass and the energy does not actually depend on $T$. We claimed that they are  conserved  quantities only for $t\in[0,T]$, since at this point we know that $\psi_t\in\dom[\F]$ only for $t\in[0,T]$. However, it is clear that, as one proves that this is true for every $t\geq0$, then one immediately extends the conservation to any $t\geq0$.
	\end{rem}
	
	\begin{rem}[Choice of the spectral parameter $ \lambda $]
		\mbox{}	\\
		The decomposition of functions in the form domain $ \dom[\F]$ defined in \eqref{eq:form_domain} depends on a spectral parameter $ \lambda > 0 $, although the domain itself is independent of $ \lambda $. In \eqref{eq:energy} we have made the choice to pick $\lambda=1$ (as suggested in \cite{A}). It is worth recalling that this is an arbitrary choice and any other choice would imply an equivalent conservation law, but a different decomposition.
	\end{rem}

	\begin{rem}[Energy form]
		\mbox{}	\\
		Another difference between the 2D case and the 1 and 3D ones is apparent in the form of the energy \eqref{eq:energy}: instead of the $ L^2$ norm of the gradient of the regular part of the wave function, \eqref{eq:energy} contains (first term) the full $ H^1 $ norm of $ \phi_{1,t} $. This is again a consequence of the impossibility to choose $ \lambda = 0 $ as a spectral parameter in the form domain decomposition. 
	\end{rem}

\subsubsection{Global well-posedness and blow-up alternative}

As we are going to see, the energy conservation is the key to prove the global well-posedness of the solution for $ \beta_0>0$. On the opposite, in the focusing case, i.e., if $ \beta_0 < 0 $, the solution might be non-global due to a blow-up at finite time. It is important to remark that, unlike the NLS with concentrated nonlinearity in 3D, one expects that no critical power occurs in the 2D focusing case and hence, as soon as $ \beta_0< 0  $, a blow-up solution might show up \cite{A}. We plan to deal with this question in a forthcoming paper.

	\begin{teo}[Global well-posedness]
		\label{teo:global}	
		\mbox{}	\\
		Let $ \sigma \geq \frac{1}{2} $ and $ \beta_0 > 0 $. Then, the solution to \eqref{eq:cauchyweak} provided by $\psi_t$ defined by \eqref{eq:ansatz} and \eqref{eq:charge_eq} is global in time, for any initial datum $ \psi_0 \in \dom $.
	\end{teo}

As anticipated, in the focusing case we have a blow-up alternative:

	\begin{pro}[Blow-up alternative]
		\label{pro:blow-up}
		\mbox{}	\\
		Let $ \sigma \geq \frac{1}{2} $, $ \beta_0 < 0 $ and $ \psi_0 \in \dom $. Then, the solution to \eqref{eq:cauchyweak} provided by $\psi_t$ defined by \eqref{eq:ansatz} and \eqref{eq:charge_eq} is  either global in time or it blows-up in a finite time.
	\end{pro}
	
	\begin{rem}[Behavior as $ t \to + \infty $]
		\mbox{}	\\
		In fact the proofs of the Theorem \ref{teo:global} and Proposition \ref{pro:blow-up} provide more information than what is contained in the statements. Indeed, while in the defocusing case $ \beta_0 > 0 $, the charge $ q(t) $ is uniformly bounded, i.e., $ \limsup_{t \to + \infty} |q(t)| < + \infty $, in the focusing one, i.e., if $ \beta_0 < 0 $, the global existence of the solution does not imply its boundedness at $ \infty $. More precisely it may happen that the maximal existence time for $ q(t) $ is $ + \infty $ but $ \limsup_{t \to + \infty} |q(t)| = + \infty $.		
	\end{rem}

	\bigskip
	
\noindent
{\bf Acknowledgements.} The authors acknowledge the support of MIUR through the FIR grant 2013 ``Condensed Matter in Mathematical Physics (Cond-Math)'' (code RBFR13WAET). We also thank \textsc{A. Fiorenza} (Universit\`{a} ``Federico II'' di Napoli) and \textsc{A. Teta} (Universit\`{a} degli Studi di Roma ``La Sapienza'') for fruitful discussions about the topic of the paper. {Finally, we acknowledge the anonymous referee of the paper whose comments and remarks lead to a major improvement of the work}.


\section{Proofs}
\label{sec:proofs}

This Sect. is devoted to the proofs of our main results. We divide this section in five steps:
\ben[(i)]
 \item we point out in Sect. \ref{sec:preliminaries} some relevant properties of Sobolev spaces of fractional index and of the integral operator $ I $;
 \item we present in Sect. \ref{sec:charge} a justification for the ansatz \eqref{eq:ansatz_pre} and the charge equation \eqref{eq:charge_eq_pre};
 \item we prove existence, uniqueness and regularity of the solution of \eqref{eq:charge_eq} and show how this allows to prove Theorem \ref{teo:local} (Sect. \ref{sec:local});
 \item we prove in Sect. \ref{sec:conservation} mass and energy conservation (Theorem \ref{teo:conservation});
 \item we use the conservation laws of Sect. \ref{sec:conservation} to prove global existence and blow-up alternative (Theorem \ref{teo:global} and Proposition \ref{pro:blow-up}).
\een

 We stress that the proof strategy differs very much from the one followed in 1 or 3D. In those cases the core of the argument heavily relies on the regularizing properties of the Abel operator, which is involved in the integral version of the charge equation. Such an operator guarantees the minimal amount of regularity on $ q(t) $ needed to ensure that the ansatz $ \psi_t $ solves the weak problem \eqref{eq:cauchyweak}, at least if the initial datum is regular enough. Unfortunately, the 2D analogue of the Abel operator is the integral operator $ I $, which does not provide any improvement of regularity (see Lemma \ref{contr_lemma}). Therefore the strategy itself of the proof needs to be modified: the required regularity of $ q(t) $ is indeed obtained by applying a suitable contraction argument to the charge equation. There are however  some drawbacks in this approach, taking the form of additional conditions on the initial state, i.e., $ \psi_0 \in \dom $, and on the nonlinearity exponent, i.e., $ \sigma \geq 1/2 $.
 

\subsection{Preliminary results}
\label{sec:preliminaries}

We start by recalling briefly some facts on Sobolev spaces with fractional index. Let $-\infty \leq a < b \leq +\infty$ and $\nu \in (0,1)$, we denote by $H^\nu(a,b)$ the Sobolev space defined by
\[
 H^\nu(a,b) = \lf\{f\in L^2(a,b) \: \big| \: \lf[f\ri]_{\dot{H}^\nu(a,b)}^2<\infty \ri\},
\]
where
\[
 \lf[f\ri]_{\dot{H}^\nu(a,b)}^2 := \int_{[a,b]^2}\dt\dtau\:\frac{|f(t)-f(\tau)|^2}{|t-\tau|^{1+2\nu}}.
\]
The space $ H^\nu(a,b) $ is a Hilbert space with the natural norm
\beq
	\label{eq:hnu norm}
 \lf\|f\ri\|_{H^\nu(a,b)}^2 = \lf\|f\ri\|_{L^2(a,b)}^2 + \lf[f\ri]_{\dot{H}^\nu(a,b)}^2.
\eeq
When $a = -\infty$ and $b = +\infty$, $H^\nu(\R)$ can be equivalently defined using the Fourier transform $ \widehat{f} $ of $ f $ and, for any $f\in L^2(\R^d)$, we will use the following convention
\begin{equation}
 \label{eq:fourier}
 \widehat{f}(\pv) : =\frac{1}{(2\pi)^{d/2}}\int_{\R^d} \diff \tv \:e^{-i \pv \cdot \tv}\,f(\tv).
\end{equation}
Consistently, the convolution of two functions $ f, g \in L^2(\R^d) $ is defined as
\[
  	\lf(f*g\ri)(\xv) : = \frac{1}{(2\pi)^{d/2}} \int_{\R^d} \diff \yv \: f(\xv - \yv) g(\yv).
\]

We start by discussing a technical point concerning the extension of functions in $ H^{1/2}(0,T) $: it is known that if $ f \in H^{\nu}(0,T) $, for $ \nu < 1/2 $ (also for $ \nu > 1/2 $ but with the additional assumption that $ f(0) = f(T) = 0 $), then $ \one_{[0,T]}(t) f(t) \in H^{\nu}(\R) $ (see, e.g., \cite[Lemma 2.1]{CFNT}). However, the case $ \nu = 1/2 $ is very special and not included in the above result because the Hardy inequality, which is a key ingredient of the proof, fails in $ H^{1/2} $ (see \cite{KP}). In the Proposition below we show that if $ f \in H^{1/2}(0,T) $ is continuous and satisfies an additional condition, then the extension to an $H^{1/2} $ function of the real line supported on a compact set is possible. We introduce an ad hoc space of continuous functions {(see also, e.g., \cite{CUF})}: for $ \beta > 0 $ we set
\bml{
  		{C_{\log,\beta}[0,T]: = \Big\{ f \in C[0,T] \: \big| \: \exists C >0 \text{ s.t. } \forall t\in[0,T], \: \exists\delta>0 \text{ s.t. }} \\
  		{\forall s \in (t - \delta, t + \delta) \cap [0,T], \: \left| f(t) - f(s) \right| \leq C \left|\log|t - s|\right|^{-\beta} \Big\}}.
}
Hence, functions in $ \clog $ satisfies a sort of {local} ``weak'' H\"{o}lder continuity condition, which is going to play a very important role in the proposition below.

	\begin{pro}[Extension of functions in $ \clog $]
		\label{pro:extension}
		\mbox{}	\\
		Let $ T > 0 $ and $ \beta > 1/2 $, then for any $ f(t) \in \clog[0,T] \cap H^{1/2}(0,T) $ with $ f(T) = 0 $, the function
		\[
			\fe(t) : = 
			\begin{cases}
				f(t),	&	\mbox{if } t \in [0,T],	\\
				f(-t),	&	\mbox{if } t \in [-T,0],	\\
				0,		&	\mbox{otherwise},
			\end{cases}
		\]
		belongs to $ H^{1/2}(\R) $.
	\end{pro}
	
	\begin{proof}
		The function $ \fe $ is obtained from $ f $ by reflecting it in an even way, so that $ \mathrm{supp}(\fe) = [-T,T] $. Of course $ \fe(t) = f(t) $ for $ t \in [0,T] $ and
		\bdm
			\lf\| \fe \ri\|_{L^{2}(-T,T)}^2 = 2 \lf\| f \ri\|_{L^{2}(0,T)}^2,	\qquad		\lf[ \fe \ri]_{\dot{H}^{1/2}(-T,T)}^2 \leq 4 \lf[ f \ri]_{\dot{H}^{1/2}(0,T)}^2,
		\edm
		and therefore, if $ f \in H^{1/2}(0,T) $, then $ \fe \in H^{1/2}(-T,T) $. Also $ \lf\| \fe \ri\|_{L^2(\R)} = \lf\| \fe \ri\|_{L^2(-T,T)} $. Now a simple computation yields
		\[
			\lf[ \fe \ri]_{\dot{H}^{1/2}(\R)}^2 = \lf[ \fe \ri]_{\dot{H}^{1/2}(-T,T)}^2 + 2 \int_{-T}^T \diff t \: \lf( \frac{1}{t+T} + \frac{1}{T-t} \ri) |\fe(t)|^2,
		\]
		and, if we show that the second term on the r.h.s. is finite, then we complete the proof. A direct inspection of those integrals reveals that the integrand is an integrable function with possibly some singularity at the boundary of the domain, where we have to verify that it still is integrable. This request can be easily seen to be that
		$$
			\frac{|\fe(t)|^2}{T-t} = \frac{|f(t)|^2}{T -t} 
		$$
		is integrable at $ t = T $. However, since by assumption $ f(T) = 0 $, the fact that $ f \in \clog[0,T] $ implies that, for $ t $ in a neighborhood of $ T  $,
		$$
			\lf| f(t) \ri| \leq \frac{C}{|\log(T-t)|^{\beta}},
		$$
		for some $ \beta > 1/2 $. Hence 
		$$
			 \frac{|\fe(t)|^2}{T-t} \leq  \frac{C}{(T - t)|\log(T-t)|^{2\beta}},
		$$
		which is integrable close to $ t = T $.
	\end{proof}

Another useful result we prove is about the Lipschitz continuity of the map $f\mapsto|f|^{2\sigma}f$ w.r.t. to the $ H^\nu $ and $ L^{\infty}  $ norm. Such a result will play an important role when inspecting the regularity of the solution to the charge equation.

\begin{lem}[Lipschitz continuity of $f\mapsto|f|^{2\sigma}f$]
 \label{lem-lip}
 \mbox{}	\\
 Let $\sigma\geq\frac{1}{2}$, $\nu \in [0,1]$ and $T,\,M > 0$. Assume also that $f$ and $g$ are functions satisfying
 \begin{equation}
  \label{eq-lip_ass}
  \lf\|f\ri\|_{L^\infty(0,T)}+ \lf\|f\ri\|_{H^\nu(0,T)} \leq M, \qquad	 \lf\|g\ri\|_{L^\infty(0,T)}+ \lf\|g\ri\|_{H^\nu(0,T)}\leq M.
 \end{equation}
 Then, there exists a constant $C>0$ independent of $f,\,g,\,M\mbox{ and }T$, such that
 \begin{equation}
  \label{eq:lip_inf}
  \lf\||f|^{2\sigma}f-|g|^{2\sigma}g \ri\|_{L^\infty(0,T)} \leq CM^{2\sigma} \lf\|f-g \ri\|_{L^\infty(0,T)}
 \end{equation}
 and
 \begin{equation}
  \label{eq:lip_sob}
  {\lf\||f|^{2\sigma}f-|g|^{2\sigma}g \ri\|_{H^\nu(0,T)} \leq C\max\lf\{1,\sqrt{T}\ri\} M^{2\sigma}\left(\lf\| f-g \ri\|_{L^\infty(0,T)} + \lf\| f-g \ri\|_{H^\nu(0,T)}\right).}
 \end{equation}
\end{lem}

\begin{proof}
	Let us first focus on \eqref{eq:lip_inf}: denote by $\varphi:\C\to\C$ the function $\varphi(z)=|z|^{2\sigma}z$. For $\sigma\geq\frac{1}{2}$, $\varphi\in C^2(\R^2;\C)$, as a function of the real and imaginary parts of $ z $. Moreover for $ z_1, z_2\in\C$,
 \begin{equation}
  \label{eq-diff_repr}
  \varphi(z_1)-\varphi(z_2)=(z_1-z_2)\,\psi_1(z_1,z_2)+(z_2-z_1)^*\,\psi_2(z_1,z_2),
 \end{equation}
 with
 \[
  \psi_1(z_1,z_2)=\int_0^1\ds\:\partial_z\varphi(z_1+s(z_2-z_1)),\qquad \psi_2(z_1,z_2)=\int_0^1\ds\:\partial_{z^*}\varphi(z_1+s(z_2-z_1)).
 \]
 where $\partial_z\varphi = (\sigma+1)|z|^{2\sigma} $ and $\partial_{z^*}\varphi=\sigma|z|^{2(\sigma-1)}z^2$. Consequently, 
 \begin{equation}
  \label{eq-psi_inf}
  |\psi_j(z_1,z_2)|\leq C\int_0^1\ds\:|z_1+s(z_2-z_1)|^{2\sigma},\quad j=1,2.
 \end{equation}
 Thus,
 \begin{equation}
  \label{eq:help}
  |\varphi(z_1)-\varphi(z_2)|\leq C\max\{|z_1|,|z_2|\}^{2\sigma}\lf|z_1-z_2\ri|
 \end{equation}
 and, then, setting $z_1=f(t)$ and $z_2=g(t)$, \eqref{eq-lip_ass} immediately entails \eqref{eq:lip_inf}.
 
 Let us now consider \eqref{eq:lip_sob}. Setting again $z_1=f(t)$ and $z_2=g(t)$ in \eqref{eq-diff_repr}, we have that
 \[
  \varphi(f(t))-\varphi(g(t))=(f(t)-g(t))\,\psi_1(f(t),g(t))+(f(t)-g(t))^*\,\psi_2(f(t),g(t)).
 \]
 For any pair of functions $ f_1, f_2 \in  H^\nu(0,T) \cap L^{\infty}(0,T) $
 \bdm
 	{\lf\| f_1 f_2 \ri\|_{H^{\nu}(0,T)} \leq C\left(\|f_1\|_{L^\infty(0,T)}\|f_2\|_{H^\nu(0,T)}+\|f_2\|_{L^\infty(0,T)}\|f_1\|_{H^\nu(0,T)}\right),}
\edm
as it can be easily seen by exploiting \eqref{eq:hnu norm}. Hence, since by \eqref{eq-lip_ass} and \eqref{eq-psi_inf},
 \[
  \lf\|\psi_j(g(t),f(t))\ri\|_{L^\infty(0,T)}\leq C M^{2\sigma},\quad j=1,2,
 \]
then, denoting $ \phi_j(t) : = \psi_j(f(t),g(t)) $ for short,
\bmln{
	\lf\|\varphi(f(t))-\varphi(g(t)) \ri\|_{H^\nu(0,T)} \leq \lf\| \phi_1  \cdot \lf( f - g \ri) \ri\|_{H^\nu(0,T)} + \lf\| \phi_2  \cdot \lf( f - g \ri)  \ri\|_{H^\nu(0,T)}	\\[.2cm]
	{\leq C \max\left\{M^{2\sigma},\left[ \lf\| \phi_1  \ri\|_{H^\nu(0,T)} + \lf\| \phi_2  \ri\|_{H^\nu(0,T)} \right]\right\} \left(\lf\| f-g \ri\|_{L^\infty(0,T)}+\lf\| f-g \ri\|_{H^\nu(0,T)}\right).}
}
Therefore it remains to verify that $ \phi_j \in H^{\nu}(0,T) $ and estimate its norm: the $ L^2 $ norm of $ \phi_j $ can be bounded as
 \begin{equation}
  \label{eq-psi_ldue}
  \lf\|\phi_j\ri\|_{L^2(0,T)}\leq C\sqrt{T}\,M^{2\sigma},\quad j=1,2.
 \end{equation}
 Hence, it is left to prove that the semi-norms are also bounded. To this aim one notes that, for fixed $z_1,z_2,w_1,w_2\in\C$, one can write
\beq
	\label{eq-diff_psi}
  	\psi_j(z_2,w_2)-\psi_j(z_1,w_1) = \psi_j(z_2,w_2) - \psi_j(z_2,w_1) + \psi_j(z_2,w_1) - \psi_j(z_1,w_1),
\eeq
and, arguing as before,
\bml{
	\label{eq-diff_psi1}
	 \psi_j(z_2,w_2) - \psi_j(z_2,w_1) = (w_2-w_1) \int_0^1 \diff s \:\partial_z \chi_j(w_1+ s(w_2-w_1)) 	\\
	 + (w_2 - w_1)^* \int_0^1 \diff s \:\partial_{z^*} \chi_j(w_1+ s(w_2-w_1)),
}
where we have set $ \chi_j(z) : = \psi_j(z_2,z) $. Similarly
\bml{
	\label{eq-diff_psi2}
	 \psi_j(z_2,w_1) - \psi_j(z_1,w_1) = (z_2-z_1) \int_0^1 \diff s \:\partial_z \xi_j(z_1+ s(z_2-z_1)) 	\\
	 + (z_2 - z_1)^* \int_0^1 \diff s \:\partial_{z^*} \xi_j(z_1+ s(z_2-z_1)),
}
with $ \xi_j(z) : = \psi_j(z,w_1) $. Now, since 
\bdm
	\chi_{1,2}(z) = \int_0^1 \diff s \: \partial_{z/z^*}\varphi(z_2+s(z-z_2)),	\qquad		\xi_{1,2}(z) = \int_0^1 \diff s \: \partial_{z/z^*}\varphi(z+s(z-w_1))	,
\edm 
and
\beqn
	\partial_z^2 \varphi (z) &=& \sigma (\sigma+1) \lf|z\ri|^{2(\sigma - 1)} z^*,	\nonumber \\
	\partial_{z} \partial_{z^*} \varphi(z) &=& \sigma (\sigma +1) \lf|z\ri|^{2(\sigma - 1)} z,	\nonumber \\			\partial_{z^*}^2 \varphi (z) &=& \sigma (\sigma-1) \lf|z\ri|^{2(\sigma - 2)} z^3,	\nonumber 
\eeqn
plugging \eqref{eq-diff_psi1} and \eqref{eq-diff_psi2} into \eqref{eq-diff_psi}, one sees that
\beq
	\label{eq-diff_psi3}
  \lf|\psi_j(z_2,w_2)-\psi_j(z_1,w_1) \ri| \leq C \max\{|z_1|,|z_2|,|w_1|,|w_2|\}^{2\sigma-1}\left(|z_2-z_1|+|w_2-w_1|\right),
\eeq
which yields
 \[
  \lf[\phi_j\ri]_{\dot{H}^\nu(0,T)} = \lf[\psi_j(f(t),g(t))\ri]_{\dot{H}^\nu(0,T)} \leq CM^{2\sigma-1}\left(\lf[f\ri]_{\dot{H}^\nu(0,T)}+ \lf[g\ri]_{\dot{H}^\nu(0,T)}\right)\leq CM^{2\sigma}.
 \]
 Thus, combining with \eqref{eq-psi_ldue},
 \[
  \lf\|\psi_j(f(t),g(t)) \ri\|_{H^\nu(0,T)}\leq C\max\lf\{1,\sqrt{T}\ri\}M^{2\sigma}, 
 \]
 so that \eqref{eq:lip_sob} is proved.
\end{proof}

\begin{rem}(Condition $ \sigma \geq \frac{1}{2} $)
 	\label{eq:ass_on_sigma}
	\mbox{}	\\
 We stress that assuming $\sigma\geq \frac{1}{2}$ is crucial in the proof \eqref{eq:lip_sob}, in particular when assuming that $ |z|^{2\sigma}z \in C^2(\R^2;\C)$ or, equivalently, in assuring that the exponent $ 2 \sigma - 1 $ in \eqref{eq-diff_psi3} is positive and therefore the functions $ f(t) $ and $ g(t) $ can be replaced in the upper bound with their suprema. 

On the other hand, \eqref{eq:lip_inf} only requires $|z|^{2\sigma}z\in C^1(\R^2;\C)$ and hence is valid for $\sigma\geq0$. In fact the estimate \eqref{eq-psi_inf} holds true for $ \sigma \geq 0 $ as well, but the stricter request $ \sigma \geq \frac{1}{2} $ enters into the derivation of the bounds on the $H^\nu$-norm of  $ \phi_j $, as explained above.
\end{rem}

In the second part of the section, we investigate some properties of the integral operator $\I$  associated with the Volterra function of order $-1$, defined by \eqref{eq:i}, i.e.,
\begin{equation}
 \label{eq:integral_i}
 (If)(t)  := \int_0^t \dtau \: \I (t - \tau) f(\tau).
\end{equation}
First, we recall some basic properties of $ \I(t) $ (for further details we refer to \cite[Sec. 18.3]{E}, where $ \I(t) $ is denoted by $ \nu(t,-1) $). The asymptotic expansions of $ \I(t) $ as $ t \to 0 $ and $ t \to \infty $ are
\begin{equation}
 \label{eq:Iasympt}
 \I(t) \underset{t \to 0}{=} \frac{1}{t \log^2 \left(\frac{1}{t}\right)}\left[1 + \mathcal{O}(\left|\log t \right|^{-1}) \right],
\end{equation}
\[
 \I(t) \underset{t \to \infty}{=} e^{t}+\mathcal{O}(t^{-1} ).
\]
Since $\I$ is continuous for $t>0$ the previous expansions entail that
\bdm
	\I(t) \in L^{1}_{\textrm{loc}}(\R^+) \cap L^{\infty}_{\mathrm{loc}}(\R^+\setminus\{0 \}) 
\edm
Furthermore, it is also worth to point out some features of the function $\NN$, defined as
\begin{equation}
 \label{eq:N}
 \NN(t) :=\int_0^t\dtau\: \I(\tau).
\end{equation}
Clearly the fact that $ \I(t) \in L^{1}_{\textrm{loc}}(\R^+)$ implies that the function $\NN$ is absolutely continuous on any bounded interval $[0,T]$, $T>0$, and $\NN(0)=0$. In addition, as $\I$ is strictly positive, $\NN$ is strictly increasing on $[0,\infty)$ and the asymptotic expansion as $ t \to 0 $ is
\begin{equation}
 \label{eq:Nasympt}
 \NN(t) \underset{t \to 0}{=} \int_{-\infty}^{\log t} \diff x \: \frac{1}{x^2} \lf(1 + \OO(x^{-1})  \ri) = \frac{1}{\log\left(\tfrac{1}{t}\right)}+\mathcal{O}(|\log t|^{-2}).
\end{equation}
Another important property of $ \NN(t) $ is stated in the next

\begin{lem}
	\label{lem:N}
 	\mbox{}	\\
 	Let $\NN(t) $ be defined in \eqref{eq:N}, then, for any $ T>0 $, $\NN(t) \in H^\nu(0,T)$, $ \forall \nu\in \lf[0,\frac{1}{2} \ri]$, and
 	\beq
 		\label{eq:Nto0}
 		\lim_{T \to 0} \lf\| \NN \ri\|_{H^{\nu}(0,T)} = 0.
	\eeq
\end{lem}

\begin{proof}
The absolute continuity of $ \NN(t) $ in the interval $ [0,T] $ implies that $\NN\in L^2(0,T)$, for any finite $ T $. Consequently it is left to prove that the seminorm $ \lf[\NN \ri]_{\dot{H}^{1/2}(0,T)} $ in bounded. An easy computation shows that
 \[
  	\lf[ \NN \ri]_{\dot{H}^{1/2}(0,T)}^2  = 2\int_0^T\dt\int_0^{\frac{t}{2}}\ds\:\left|\frac{\NN(t)-\NN(s)}{t-s}\right|^2+2\int_0^T\dt\int_{\frac{t}{2}}^t\ds\:\left|\frac{\NN(t)-\NN(s)}{t-s}\right|^2.                              
 \]
 Looking at the first integral and recalling that $\NN$ is increasing, we find
 \[
  \left|\frac{\NN(t)-\NN(s)}{t-s}\right|^2\leq4\,\frac{\NN^2(t)}{t^2},	\quad	\forall s\in \lf(0,\tfrac{t}{2}\ri).
 \]
 Hence,
 \[
  \int_0^T\dt\int_0^{\frac{t}{2}}\ds\:\left|\frac{\NN(t)-\NN(s)}{t-s}\right|^2\leq 2 \int_0^T\dt\:\frac{\NN^2(t)}{t} < \infty,
 \]
 since, by \eqref{eq:Nasympt}, $\frac{\NN^2(t)}{t}\sim\I(t)$, when $t\to0$, and thus is integrable over $[0,T]$, for any $ T $ finite.
 
Applying Cauchy inequality to the second integral, we get
\bmln{
  \int_0^T\dt\int_{\frac{t}{2}}^t\ds\:\left|\frac{\NN(t)-\NN(s)}{t-s}\right|^2 = \int_0^T\dt\int_{\frac{t}{2}}^t\ds\:\left|\frac{1}{t-s}\int_s^t\dtau\:\I(\tau)\right|^2	\\
                                                                               \leq \int_0^T\dt\int_{\frac{t}{2}}^t\ds\:\frac{1}{t-s}\int_s^t\dtau\:\I^2(\tau).
}
 Furthermore since $\I$ is positive and convex (see \cite{CF}), it is  $\I^2(\tau)\leq \I^2(t)+\I^2(s)$ for every $\tau\in[s,t]$, so that
 \[
  \int_0^T\dt\int_{\frac{t}{2}}^t\ds\:\frac{1}{t-s}\int_s^t\dtau\:\I^2(\tau)\leq \int_0^T\dt\:\int_{\frac{t}{2}}^t\ds\:(\I^2(t)+\I^2(s)).
 \]
 Now, noting that $\log^{-4}(1/s)\leq\log^{-4}(1/t)$ for all $s\in(t/2,t)$ and using again \eqref{eq:Iasympt},
 \[
  \int_0^T\dt\int_{\frac{t}{2}}^t\ds\:\I^2(s) \sim \int_0^T\dt\int_{\frac{t}{2}}^t\ds\:\frac{1}{s^2\log^4(\frac{1}{s})} \leq C\int_0^T\dt\:\frac{1}{t\log^4(\frac{1}{t})}<\infty,
 \]
 whereas, on the other hand, 
 \[
  \int_0^T\dt\int_{\frac{t}{2}}^t\ds\:\I^2(t)\leq C \int_0^T\dt\:t\,\I^2(t)\sim C\int_0^T\dt\:\frac{1}{t\log^4(\frac{1}{t})}<\infty.
 \]
 Thus
 \[
  \int_0^T \diff t \int_{\frac{t}{2}}^t \diff s \: \left|\frac{\NN(t)-\NN(s)}{t-s}\right|^2 <\infty.
 \]
In conclusion we proved that $ \lf[\NN\ri]_{\dot{H}^{1/2}(0,T)} < \infty$. The same inequalities also imply \eqref{eq:Nto0}.
\end{proof}

In \cite{CF} the operator $ I $ is investigated in details and several useful properties are established. Here, we only show the most relevant ones for our application (we also mention some proofs for the sake of completeness).

	\begin{lem}
 		\label{contr_lemma_inf}
 		\mbox{}	\\
		 Let $T>0$ and $f\in L^\infty(0,T)$, then  $ If\in C[0,T]$ and 
		 \begin{equation}
			 \label{eq:est_I_inf}
			 \lf\|If\ri\|_{L^\infty(0,T)} \leq C_T \lf\|f\ri\|_{L^\infty(0,T)},
		\end{equation}
		with $C_T>0$ independent of $f$ and such that $C_T \xrightarrow[T \to 0]{} 0$.
	\end{lem}

	\begin{proof}
		Recalling \eqref{eq:integral_i} and \eqref{eq:Nasympt}, \eqref{eq:est_I_inf} is immediate. Then, it is left to prove that $If$ is continuous. To this aim, fix $t_0\in[0,T)$ and $t\in\;(t_0,T]$. Easy computations yield
 		\bdm
  I f(t)-I f(t_0) =  \int_0^T \diff \tau \: \I (t-\tau)\one_{[t_0,t]}(\tau)f(\tau)	\\
                            - \int_0^T\diff \tau \:\lf(\I(t_0-\tau)-\I(t-\tau) \ri)\one_{[0,t_0]}(\tau)f(\tau)
 		\edm
 		and therefore
 		\bml{
 			\label{eq:dominated}
 			\left |I f(t)-I f(t_0)\right|\leq \int_0^T\dtau\: \I(t-\tau)\one_{[t_0,t]}(\tau)|f(\tau)| +\int_0^T\dtau\: \lf|\I(t_0-\tau)-\I(t-\tau) \ri|\one_{[0,t_0]}(\tau)|f(\tau)| \\
 			\leq \NN(t-t_{0})\|f\|_{L^\infty(0,T)}+\|f\|_{L^\infty(0,T)} \int_{0}^{t_0}\dtau\: \lf| \I (\tau)-\I(t - t_0 + \tau) \ri|, 
 		}
 		
 		Therefore, the first term converges  to zero by the continuity of $ \NN$, while the second one tends to zero by dominated convergence. Indeed, it suffices to bound from above the integrand in the second term by an integrable function independent of $ t $. Since $ t $ varies in a bounded set and $ \I(t) $ is bounded for $ t > 0 $ finite, we have
 		\bdm
 			\lf| \I (\tau)-\I(t - t_0 + \tau) \ri| \leq \I(\tau) + \sup_{\delta \in [0,T]} \I(\tau + \delta),
		\edm
		and the r.h.s. is integrable for $ \tau \in [0,t_0] $. 		  
 		   
 		Since the same holds if $ t<t_0$, one has that $I f(t)\to I f(t_0)$ as $ t \to t_0 $, which concludes the proof.
\end{proof}

	\begin{lem}
		\label{contr_lemma}
		\mbox{}	\\
 		Let  $f\in H^{1/2}(0,T)\cap L^\infty(0,T)$, $T>0$, then $ If \in H^{1/2}(0,T)$ and, in particular, there exists $C_T>0$ independent of $f$ and satisfying $C_T \xrightarrow[T \to 0]{} 0$, such that
 		\begin{equation}
			\label{eq:est_I_lim}
			\lf\| If \ri\|_{H^{1/2}(0,T)} \leq C_T \left( \lf\| f \ri\|_{L^\infty(0,T)}+ \lf\| f \ri\|_{H^{1/2}(0,T)}\right).
 		\end{equation}
	\end{lem}

	\begin{proof}
 		Let us divide the proof in two parts: we first estimate the $ L^2$ norm of $ If $ and then the semi-norm $ \lf[ If \ri]_{H^{1/2}(0,T)} $.
 
		Let $T>0$ be finite and $f\in H^{1/2}(0,T)\cap L^\infty(0,T)$. In order to extend the operator $ I $ to an operator on the line, we set $ f_e(t) := \one_{[0,T]}(t)  f(t) $ and define 
 \[
  (I_e f)(t) : = \int_0^t\dtau \: \I_e(t-\tau)f_e(\tau), \qquad t \in \R,
 \]
 where
 \[
  	\I_e(t) := \one_{[0,T]}(t) \I(t).
 \]
 Since  $  (I_e f)(t) = (If)(t)$ for all $t\in[0,T]$,
 \begin{equation}
 	\label{eq:da_T_a_R}
 	\lf\|If \ri\|_{L^2(0,T)} = \lf\| I_e f \ri\|_{L^2(0,T)} \leq \| I_e f \|_{L^2(\R)}.
 \end{equation}
 Now, applying the Fourier transform on $ \R $ to $ I_e f $ and using the identity $$ \one_{[0,t]}(\tau) = \one_{\R^+}(\tau) - \one_{\R^+}(\tau-t), $$ one gets
 \[
  \widehat{I_e f} = \widehat{\I_e * \lf( \one_{\R^+} f_e \ri)} - \widehat{\lf( \one_{\R^-} \I_e \ri) * f_e} = \widehat{\I_e} \widehat{ \one_{\R^+} f_e } - \widehat{\one_{\R^-} \I_e} \widehat{f_e} =  \widehat{\I_e} \widehat{f_e},
 \]
since by construction $ \one_{\R^+}(t) f_e(t) = f_e(t) $ and $  \one_{\R^-}(t) \I_e(t) = 0$. Hence by \eqref{eq:da_T_a_R} and the above identity
 \[
 \lf\|I f \ri\|_{L^2(0,T)}^2 \leq \int_\R \dk \: \big| \widehat{\I_e}(k) \big|^2 \, \big|\widehat{f_e}(k) \big|^2,
 \]
but $|\widehat{\I_e}(k)| \leq C \NN(T)$ and therefore
\beq
	\label{eq:est_I_ldue}
  \lf\| If \ri\|_{L^2(0,T)}^2 \leq C \NN^2(T) \lf\|f_e \ri\|_{L^2(\R)}^2 =C \NN^2(T) \lf\| f \ri\|_{L^2(0,T)}^2,
 \eeq
 which implies the result via Lemma \ref{lem:N}.
 
 We now focus on the seminorm $ \lf[ If \ri]_{\dot{H}^{1/2}(0,T)}$. First, we note that, for every $0<s<t<T$,
 \[
  (If)(t)-(If)(s)=\int_s^t\dtau\:\I(\tau)f(t-\tau) + \int_0^s\dtau\:\I(\tau)(f(t-\tau)-f(s-\tau)),
 \]
so that
\bml{
  \label{eq:Naux_1}
  \lf[ If \ri]_{\dot{H}^{1/2}(0,T)}^2 \leq  4\int_0^T\dt\:\int_0^t\ds\:\left|\frac{1}{t-s}\int_s^t\dtau\:\I(\tau)f(t-\tau)\right|^2 \\
                                 + 4\int_0^T\dt\:\int_0^t\ds\:\left|\int_0^s\dtau\:\I(\tau)\frac{f(t-\tau)-f(s-\tau)}{t-s}\right|^2.
}
 Now, one can easily see that, since $f\in L^\infty(0,T)$,
\bml{
  	\label{eq:Naux_2}
  	4 \int_0^T\dt\:\int_0^t\ds\:\left|\frac{1}{t-s}\int_s^t\dtau\:\I(\tau)f(t-\tau)\right|^2 \leq 4\|f\|_{L^\infty(0,T)}^2\int_0^T\dt\:\int_0^t\ds\:\left|\frac{\NN(t)-\NN(s)}{t-s}\right|^2	\\
  	= 2\|f\|_{L^\infty(0,T)}^2[\NN]_{\dot{H}^{1/2}(0,T)}^2 \leq 2\|f\|_{L^\infty(0,T)}^2\|\NN\|_{H^{1/2}(0,T)}^2
}
where the last factor $\|\NN\|_{H^{1/2}(0,T)}$ is finite by Lemma \ref{lem:N}. On the other hand by Cauchy-Schwarz inequality, monotonicity of $\NN$ and positivity of $ \I $, we have
\bmln{
   \displaystyle 4\int_0^T\dt\:\int_0^t\ds\:\left|\int_0^s\dtau\:\I(\tau)\frac{f(t-\tau)-f(s-\tau)}{t-s}\right|^2	\\ 
  \displaystyle \leq4\NN(T)\int_0^T\dt\:\int_0^t\ds\:\int_0^s\dtau\:\I(\tau)\left|\frac{f(t-\tau)-f(s-\tau)}{t-s}\right|^2	\\
  \displaystyle \leq 4 \NN(T) \int_0^T\dtau\:\I(\tau)\int_0^{T-\tau}\dt\int_0^{t}\ds\left|\frac{f(t)-f(s)}{t-s}\right|^2 	\leq 2 \NN^2(T) \lf[ f \ri]_{\dot{H}^{1/2}(0,T)}^2
}
and plugging the above inequality and \eqref{eq:Naux_2} into \eqref{eq:Naux_1},
 \[
  \lf[ If \ri]_{\dot{H}^{1/2}(0,T)}^2\leq C\max\lf\{\|\NN\|_{\dot{H}^{1/2}(0,T)}^2,\NN^2(T) \ri\} \left(\lf\| f \ri\|_{L^\infty(0,T)}^2+\lf\| f \ri\|_{\dot{H}^{1/2}(0,T)}^2\right).
 \]
 Finally, the above estimate in combination with \eqref{eq:est_I_ldue} yields
 \[
  	\lf\| If \ri\|_{\dot{H}^{1/2}(0,T)}\leq C\max\lf\{\|\NN\|_{\dot{H}^{1/2}(0,T)},\NN(T) \ri\}\left( \lf\| f \ri\|_{L^\infty(0,T)}+ \lf\| f \ri\|_{\dot{H}^{1/2}(0,T)}\right)
 \]
 and, since both $\NN(T)$ and $\|\NN\|_{\dot{H}^{1/2}(0,T)}$ converges to zero as $T\to 0$ by Lemma \ref{lem:N}, the proof is complete.
\end{proof}

Finally, we point out some relevant properties of the integral operator $ J $, defined by
\begin{equation}
 \label{eq:integral_j}
 (Jf)(t) : = \int_0^t \dtau \: \J(t - \tau) f(\tau),\qquad \J(t-\tau) : = - \gamma - \log (t - \tau).
\end{equation}

	\begin{lem}
 		\label{lem:i_identity}
 		\mbox{}	\\
 		For any $ t \in \R^+ $ and $ f \in L^1(0,t) $,
 		\begin{equation}
  			\label{eq:i_inverse}
  			\lf(J I f \ri)(t) = \lf( I J f \ri)(t) = \int_0^t \dtau \: f(\tau).
 		\end{equation}
	\end{lem}

	\begin{proof}
 		{For the meaning of $Jf$, when $f$ is only an integrable function we refer the reader to \cite[proof of Theorem 5.3]{CF}: the argument is similar to the one in the first part of the proof of Lemma \ref{contr_lemma} and exploits the properties of convolutions on $ \R $, via a suitable extension of the functions involved.} 
 		
In addition, we observe that one has the identity
 		\begin{equation}
 	 		\label{eq:i_identity}
  			\int_{0}^{t} \dtau \: \I(\tau) \lf(-\gamma-\log (t - \tau) \ri) = \int_{0}^{t} \dtau \: \I(t - \tau) \lf(-\gamma-\log \tau\ri) = 1.
 		\end{equation}
 		In \cite[Lemma 32.1]{SKM} it is indeed proven that (in the formula stated in the cited Lemma  one has to take $ \alpha = 1, h = 0 $)
 		\[
  			\int_{0}^{t} \dtau \: \left(\log \tau - \psi(1) \right) \partial_t \nu(t - \tau) = -1,
 		\]
 		where $\nu$ here denotes the Volterra function of order 0. However, using \cite[Eq. (12), Sect. 18.3]{E}, one can recognize that $ \partial_t \nu(t) = \I(t) $ (and that $\psi(1)=-\gamma$).

 		Let us then prove the identity involving $ I J $. The proof of the other one is perfectly analogous and we omit it for the sake of brevity. First of all, in the expression
 		\[
  			\lf( I J f \ri)(t) = \int_0^{t} \dtau  \int_0^{t-\tau} \dsigma \: \I(\tau) \J(t-\tau-\sigma) f(\sigma),
 		\]
 		one can exchange the order of the integration, since
 		\bdm
   				\displaystyle \int_0^{t} \dtau  \int_0^{t-\tau} \dsigma \: \I(\tau) \J(t-\sigma-\tau) f(\sigma)  = \int_0^{t} \dsigma  \int_0^{t-\sigma} \dtau \: \I(\tau) \J(t-\sigma-\tau) f(\sigma).
 		\edm
 		Using \eqref{eq:i_identity}, we conclude that
 		\[
  			\lf( I J f \ri)(t) = \int_0^{t} \dsigma\:  \int_0^{t-\sigma} \dtau \: \I(\tau) \J(t-\sigma-\tau) f(\sigma) = \int_0^{t} \dsigma \: f(\sigma).
 		\]
	\end{proof}
	

\subsection{A derivation of the charge equation}
\label{sec:charge}

Before starting to discuss the charge equation, it is worth making a brief excursus on a heuristic computation, which motivates ansatz \eqref{eq:ansatz_pre} and equation \eqref{eq:charge_eq_pre}. {Note that we derive them in the case of an arbitrary number $N$ on interactions for the sake of generality}. In the following, we also assume for the sake of simplicity that $q_j(0) = 0$, for every $j=1,\dots,N$. However, one can prove that such an assumption is not restrictive.

Neglecting any regularity issue, we can compute the time derivative of \eqref{eq:ansatz_pre} and obtain that, at least formally,
\bml{
	\label{eq:time_derivative}
 	i \partial_t \psi_t (\xv) =  ( - \Delta U_0(t) \psi_0 ) (\xv) - \frac{1}{2\pi}\sum_{j=1}^N q_j(t)	+ \frac{1}{2\pi} \sum_{j=1}^N \int_0^t \dtau \:  \partial_\tau U_0 (t - \tau; |\xv - \yv_j|) \: q_j(\tau)\\
                                     =  ( - \Delta U_0(t) \psi_0 ) (\xv)  - \frac{1}{2\pi} \sum_{j=1}^N \int_0^t \dtau \:  U_0 (t - \tau; |\xv - \yv_j|) \: \dot q_j(\tau),
}
where we used the fact that (as $q_j(0)=0$) $\psi_0\in H^2(\R^2)$ and that (by definition) $i\partial_tU_0(t)\psi_0=-\Delta U_0(t)\psi_0$. Hence, applying the Fourier transform on $ \R^2 $, the above expression becomes (we set $ p = |\pv| $)
\begin{equation}
 \label{eq:t_derivative_fourier}
 i \widehat{\partial_t\psi_t}(\pv) =  p^2 e^{-ip^2 t} \widehat{\psi_0} (\pv)  - \frac{1}{2\pi} \sum_{j=1}^N \int_0^t \dtau \:  e^{- i \pv \cdot \yv_j} e^{-ip^2(t- \tau)} \: \dot q_j(\tau).
\end{equation}
The l.h.s. of \eqref{eq:time_derivative} equals the action of $ H_0 $ on the regular part of the wave function $ \psi_t $ (see \eqref{eq:lin op} and \eqref{eq:lin domain}), i.e., 
\bml{
  	\label{eq:ham_psit_fourier}
  	p^2 \bigg( \widehat{\psi_t}(\pv) - \frac{1}{2\pi} \sum_{j=1}^N \frac{q_j(t) e^{-i \pv \cdot \yv_j}}{p^2 + \lambda} \bigg) - \frac{\lambda}{2\pi} \sum_{j=1}^N \frac{q_j(t) e^{-i \pv \cdot \yv_j}}{p^2 + \lambda}\\	
  	\displaystyle = p^2  e^{-ip^2 t} \widehat{\psi_0} (\pv)  + \frac{1}{2\pi} \sum_{j=1}^N \int_0^t \dtau \:  e^{- i \pv \cdot \yv_j} \partial_\tau \left( e^{-ip^2(t- \tau)} \right) \: q_j(\tau) - \frac{1}{2\pi} \sum_{j=1}^N  q_j(t) e^{-i \pv \cdot \yv_j}\\	
  	\displaystyle =   p^2 e^{-ip^2 t} \widehat{\psi_0} (\pv)  - \frac{1}{2\pi} \sum_{j=1}^N \int_0^t \dtau \:  e^{- i \pv \cdot \yv_j} e^{-ip^2(t- \tau)} \: \dot q_j(\tau),
}
which is equal to \eqref{eq:t_derivative_fourier}. Therefore,  for any  $ \qv(t) $ and $\psi_0$ such that the r.h.s. of \eqref{eq:ham_psit_fourier} makes sense, the ansatz \eqref{eq:ansatz_pre} does solve the time-dependent Schr\"{o}dinger equation, at least in a weak sense.

Under restrictive assumptions on $ \psi_0 $, however, the ansatz $ \psi_t $ must belong to the (nonlinear) operator domain $ \dom(\ham) $, with $ \alpha_j = \beta_j |q_j(t)|^{2\sigma_j} $, $ j = 1, \ldots, N $, i.e., it must satisfy the boundary conditions \eqref{eq:boundary_condition}, which can be cast in the form
\[
 	\frac{1}{2\pi} \int_{\R^2} \dpv \: e^{i \pv \cdot \yv_j} \widehat{\phi_{\lambda,t}}(\pv) =  \left(\beta_j |q_j(t)|^{2\sigma_j} + \tfrac{1}{2\pi} \log \tfrac{\sqrt{\lambda}}{2} - \tfrac{\gamma}{2\pi} \right) q_j(t) - \frac{1}{2\pi}\sum_{k \neq j} q_k(t)  K_0\lf(\sqrt{\lambda} |\yv_j - \yv_k|\ri),
\]
In fact, as we are going to see, the above condition will force $ \qv(t) $ to be a solution to the charge equation \eqref{eq:charge_eq_pre}. Indeed, since $$\phi_{\lambda,t}=\psi_t-\frac{1}{2\pi}\sum_{k=1}^Nq_k(t)K_0 \lf(\sqrt{\lambda}|\cdot-\yv_k|\ri),$$ by \eqref{eq:ansatz_pre},
\bmln{
  	\displaystyle \frac{1}{2\pi} \int_{\R^2} \dpv \: e^{i \pv \cdot \yv_j} \bigg\{ e^{-ip^2 t} \widehat{\psi_0} (\pv) + \frac{i}{2\pi} \sum_{k=1}^N \int_0^t \dtau \:  e^{- i \pv \cdot \yv_k} e^{-ip^2(t- \tau)} \: q_k(\tau)- \frac{1}{2\pi} \sum_{k =1}^N \frac{q_k(t) e^{-i \pv \cdot \yv_k}}{p^2 + \lambda} \bigg\}\\
  	\displaystyle  = \left(\beta_j|q_j(t)|^{2\sigma_j} + \tfrac{1}{2\pi} \log \tfrac{\sqrt{\lambda}}{2} - \tfrac{\gamma}{2\pi} \right) q_j(t) - \frac{1}{2\pi}\sum_{k \neq j} q_k(t)  K_0 \lf(\sqrt{\lambda} |\yv_j - \yv_k| \ri).
}
The last off-diagonal term cancels exactly and thus the identity becomes
\bmln{
  	\displaystyle \frac{1}{2\pi} \int_{\R^2} \dpv \: e^{i \pv \cdot \yv_j} \bigg\{ e^{-ip^2 t} \widehat{\psi_0} (\pv) + \frac{i}{2\pi} \sum_{k=1}^N \int_0^t \dtau \:  e^{- i \pv \cdot \yv_k} e^{-ip^2(t- \tau)} \: q_k(\tau) - \frac{1}{2\pi} \frac{q_j(t) e^{-i \pv \cdot \yv_j}}{p^2 + \lambda} \bigg\}\\	
  	 =\displaystyle  \left( \beta_j|q_j(t)|^{2\sigma_j} + \tfrac{1}{2\pi} \log \tfrac{\sqrt{\lambda}}{2} + \tfrac{\gamma}{2\pi} \right) q_j(t).
}
Combining the last diverging term on the l.h.s. with the second one via an integration by parts (here we implicitly assume that the charge is regular enough), we get
\bmln{
  	\displaystyle \frac{1}{2\pi} \int_{\R^2} \dpv \: \bigg\{ e^{i \pv \cdot \yv_j} e^{-ip^2 t} \widehat{\psi_0} (\pv) - \frac{1}{2\pi (p^2+\lambda)} \int_0^t \dtau \:  e^{-ip^2(t- \tau)} \: \left[ \dot q_j(\tau) - i \lambda q_j(\tau) \right] \\	
 	\displaystyle + \frac{i}{2\pi } \sum_{k \neq j} \int_0^t \dtau \:  e^{ i \pv \cdot ( \yv_j - \yv_k)} e^{-ip^2(t- \tau)} \: q_k(\tau) \bigg\} = \left( \beta_j|q_j(t)|^{2\sigma_j} + \tfrac{1}{2\pi} \log \tfrac{\sqrt{\lambda}}{2} + \tfrac{\gamma}{2\pi} \right) q_j(t),
}
The $\pv$ integral of the second term on the l.h.s. contains an infrared singularity for $ t = \tau $ which is proportional to $ \log(t - \tau) $: in fact by \cite[Eqs. 3.722.1 \& 3.722.3]{GR}
\bml{
 \label{eq:sici}
 2\pi\lf(U_0(t) K_0 \big(\sqrt{\la} \cdot \big)\ri)(\mathbf{0}) = \int_{\R^2} \dpv \: \frac{e^{-i p^{2}(t-\tau)}}{p^{2}+\lambda}  = - \pi e^{i \lambda(t-\tau)} \left[ \mathrm{ci}(\lambda(t-\tau)) -i\, \mathrm{si}(\lambda(t-\tau))  \right] \\
                                                                 = - \pi e^{i \lambda (t - \tau)} \left( \gamma + \log \lambda + \log (t - \tau) \right) + e^{i \lambda (t - \tau)} Q(\lambda; t-\tau),
}
where $\textrm{si}( \: \cdot \:)$ and $ \mathrm{ci}( \: \cdot \:) $ stand for the sine and cosine integral functions \cite[Eqs. 5.2.1 \& 5.2.2]{AS} and (see, e.g., \cite[Eq. 5.2.16]{AS})
\begin{equation}
 \label{eq-Q}
 Q(\lambda; t - \tau) : = - \pi \bigg( \sum_{n=1}^{\infty}\frac{({-}(t-\tau)^{2} \lambda^{2})^{n}}{2n(2n)!} - i \, \textrm{si}((t - \tau) \lambda) \bigg)
\end{equation}
(note that $ Q(0; t - \tau) = -\tfrac{i\pi^2}{2} $). Hence, we obtain
\bmln{
	 \left(U_0(t) \psi_0\right)(\yv_j) + \frac{i}{2\pi } \sum_{k \neq j} \int_0^t \dtau \:  U_0(t-\tau; |\yv_j - \yv_k|) \: q_k(\tau) - \left( \beta_j|q_j(t)|^{2\sigma_j} + \tfrac{1}{2\pi} \log \tfrac{\sqrt{\lambda}}{2} + \tfrac{\gamma}{2\pi} \right) q_j(t)   \\
 	= - \frac{1}{4 \pi} \int_0^t \dtau \: \left( \gamma +\log (t - \tau) + \log\lambda  - \tfrac{1}{\pi} Q(\lambda; t-\tau) \right)  \partial_{\tau} \left( e^{i \lambda (t - \tau)} q_j(\tau) \right) 
}
and taking the limit $ \lambda \to 0 $ (notice the exact cancellation of the diverging $ \log \lambda $ terms)
\bmln{
 	(U_0(t) \psi_0)(\yv_j) + \frac{i}{2\pi } \sum_{k \neq j} \int_0^t \dtau \:  U_0(t-\tau; |\yv_j - \yv_k|) \: q_k(\tau)	\\
		 - \left(\beta_j|q_j(t)|^{2\sigma_j} - \tfrac{1}{2\pi} \log 2 + \tfrac{\gamma}{2\pi} -\tfrac{i}{8} \right) q_j(t)= - \frac{1}{4 \pi} \int_0^t \dtau \: \left( \gamma +\log (t-\tau)    \right)  \dot q_j(\tau).
}
If we now apply to both sides the integral operator $ I $ defined in \eqref{eq:integral_i} and exploit the property proven in Lemma \ref{lem:i_identity}, we find
\bmln{
  \displaystyle \int_0^t \dtau \: \I(t-\tau) (U_0(\tau)\psi_0)(\yv_j) + \frac{i}{2\pi}\sum_{k\neq j}\int_0^t \dtau \: \I(t-\tau)\int_0^\tau \dmu \: U_0(\tau-\mu;|\yv_j-\yv_k|)q_k(\mu) \\
  \displaystyle - \int_0^t \dtau \: \I(t-\tau) \beta_j|q_j(\tau)|^{2\sigma_j}q_j(\tau) + \frac{1}{2\pi}\left(\log 2 - \gamma + \tx\frac{i\pi}{4}\right)\int_0^t \dtau \: \I(t-\tau)q_j(\tau) = \frac{q_j(t)}{4\pi}
}
{and thus multiplying each term by $4\pi$ and exchanging the integration order in the sum one obtains \eqref{eq:charge_eq_pre}.}

{This formal derivation deserves two further comments. The first one concerns the limit $\lambda\to0$ that one performs in order to arrive to the actual form of the charge equation. The two issues concerning this point are both that the choice $\lambda=0$ is forbidden in the 2D case (for the domain decomposition) and that the charge equation must be independent of the choice of $\lambda$. However, using the Laplace transform (see, e.g., \cite{A}) one can check (with long computations) that the same form of the charge equation can be obtained for any other choice of $\lambda>0$.}

{Finally, it is worth giving some further details on the reason why the off-diagonal terms cannot be treated in this paper. As we stressed at the beginning of Section \ref{sec:mainres}, the point is that the kernel of the off-diagonal terms, i.e.,
\begin{equation}
\label{eq:problema}
\int_0^t\dtau\:\I(t-\tau)U_0(\tau;|\yv_j-\yv_k|)
\end{equation}
is not integrable. The mathematical reason for such a lack of integrability is the following. Using the asymptotics provided by \eqref{eq:Iasympt}, one can check that, up to some constants, the local singularity of \eqref{eq:problema} is given by
\begin{equation}
\label{problema2}
 \frac{1}{t\log^2t}\int_0^1\dtau\:\frac{e^{i\frac{|\yv_j-\yv_k|^2}{4t\tau}}}{\tau(1 - \tau)} \frac{1}{\left(1+\frac{\log(1-\tau)}{\log t}\right)^2}.
\end{equation}
If we split the integral in two parts, one can see that the integral between 0 and $\frac{1}{2}$ converges to a constant when $t\to0$, whereas the integral between $\frac{1}{2}$ and 1 yields a logarithmic divergence as $t\to0$, thus implying that \eqref{problema2} behaves like $ \frac{1}{t \log t} $ as $ t \to 0 $ and therefore it is not integrable.}

{As a further remark, one could note that, in fact, the term which is logarithmic-divergent is multiplied by a factor that is strongly oscillating as $t\to0$. This, although not relevant concerning the integrability of the kernel, could provide nevertheless the possibility of studying this types of 
kernels, but the corresponding theory is yet to be developed.}


\subsection{Local well-posedness}
\label{sec:local}

In order to prove Theorem \ref{teo:local}, the first step is to discuss existence, uniqueness and Sobolev regularity of any solution of \eqref{eq:charge_eq}. We split the results into two separate Propositions to make the proof strategy clearer: by general arguments about Volterra-type integral equations and the properties of \eqref{eq:charge_eq}, we obtain existence and uniqueness of a continuous solution $ q(t) $ up to some maximal existence time $ T_* $, which might as well be $ +\infty $. Then, in order to derive the Sobolev regularity of $ q(t) $, we use the aforementioned contraction, which works on some a priori shorter interval $ [0,T] $, $ T < T_* $. In Proposition \ref{pro:extension q} we will however show how one can extend such a regularity to the whole existence interval, provided a property of the source term holds true.

Preliminarily, note that \eqref{eq:charge_eq} can be written in a compact form as
\begin{equation}
 \label{eq:charge_compact}
 q(t) + \int_0^t \diff \tau \: \bigg(g(t,\tau,q(\tau))+\kappa\I(t-\tau) \: q(\tau)\bigg) = f(t),
\end{equation}
where $\kappa:=-2\big(\log 2-\gamma+i\frac{\pi}{4}\big)$ and $g$ and $f$ are defined respectively by
\beq
 g(t,\tau,q(\tau))=4\pi\beta_0\I(t-\tau)|q(\tau)|^{2\sigma}q(\tau),
\eeq
\begin{equation}
 \label{eq:mapf}
 f(t)=4\pi\int_0^t\dtau\:\I(t-\tau)(U_0(\tau)\psi_0)(\yv).
\end{equation}

	\begin{pro}[Continuity of $ q(t) $]
		\label{pro:charge continuous}
 		\mbox{}	\\
 		Let $\sigma\geq\frac{1}{2}$ and $  \psi_0\in \dom $. Then, there exists $T_{*} >0$ such that \eqref{eq:charge_compact} admits a unique solution $ q(t) \in C[0,T_*) $. Moreover either $ T_{*} = +\infty $, i.e., the solution is global in time, or $ T_* < + \infty $ and $ \lim_{t \to T_*} \lf| q(t) \ri| = + \infty $.
	\end{pro}
	
	{\begin{rem}
	 The assumptions $\sigma\geq\frac{1}{2}$ and $\psi\in\dom$ are not necessary in the above statement: indeed, everything works as well even if we require only that $\sigma\geq0$ and $\widehat{\phi_{\lambda,0}}\in L^1(\R^2)$. However, since the main results of the paper request those more limiting assumptions, we keep them even in the intermediate results, in order to not give rise to misunderstandings. Analogous considerations hold for Proposition \ref{pro:charge} concerning the smoothness of the regular part of the initial datum.
	\end{rem}}

	\begin{proof}
	 	The result is obtained by directly applying {\cite[Corollary 2.7]{M}}: it claims that there exists $T_*>0$ for which \eqref{eq:charge_compact} admits a unique solution $q\in C[0,T^*)$, with the claimed properties, provided
		\ben[(i)]
			\item 	$ f $ is continuous on $\R^+ $;
 			\item 	for every $t'>0$ and every bounded set $\BB \subset\C$, there exists a measurable function $m(t,\tau)$ such that
					\[
					\lf| g(t,\tau,q)+\kappa\I(t-\tau)q \ri| \leq m(t,\tau),\quad \forall\: 0\leq\tau\leq t\leq t', \quad \forall q\in \BB,
					\]
 					with 
 					\bdm
 						\sup_{t \in [0, t']} \int_0^t\dtau\:m(t,\tau) < \infty, \qquad	\int_0^t\dtau\:m(t,\tau) \xrightarrow[t \to 0]{} 0;
					\edm
 			\item 	for every compact interval $I\subset\R^+$, every continuous function $\varphi:I\to\C$ and every $t_0\in\R^+$,
					\begin{equation}
  						\label{eq-lim_I}
  						{\lim_{t\to t_0}\int_I\dtau\:\lf[ g(t,\tau,\varphi(\tau))-g(t_0,\tau,\varphi(\tau))+\kappa(\I(t-\tau)-\I(t_0-\tau))\varphi(\tau) \ri] =0;}
 					\end{equation}
 
 			\item 	for every $t'>0$ and every bounded $\BB\subset\C$, there exists a measurable function $h(t,\tau)$ such that
					\[
  						{\lf|g(t,\tau,q_1)-g(t,\tau,q_2)+\kappa\I(t-\tau)(q_1-q_2) \ri|\leq h(t,\tau) \lf|q_1-q_2 \ri|,}
 					\]
 					for all $0\leq\tau\leq t\leq t'$ and all $q_1,\,q_2\in \BB$, with $h(t,\cdot\,)\in L^1(0,t)$ for all $t\in[0,t']$ and 
					\bdm
						\int_t^{t+\ep}\dtau\:h(t+\ep,\tau) \xrightarrow[\ep \to 0]{} 0.
					\edm
		\een

		Let us now verify all the hypothesis. First, consider point (i): since $ \psi_0 \in \dom[\F] $,
		\begin{equation}
			\label{eq:Aunoduetre}
  			4\pi(U_0(\tau)\psi_0)(\yv)=\underbrace{4\pi \lf(U_0(\tau)\phi_{\lambda,0} \ri)(\yv)}_{:=A_1(\tau)} + \underbrace{2q(0) \lf(U_0(\tau)K_0\lf(\sqrt{\lambda}|\cdot -\yv|\ri) \ri)(\yv)}_{:=A_2(\tau)}.
		\end{equation}
 		Observing that
 		\[
  			A_1(\tau) = 2 \int_{\R^2} \dpv \: e^{i\pv \cdot \yv} e^{-i p^2\tau} \widehat{\phi_{\lambda,0}}(\pv)
 		\]
 		and recalling that $\widehat{\phi_{\lambda,0}}\in L^1(\R^2)$ by assumption, one sees that $A_1$ is bounded and therefore $IA_1$ is continuous as well by Lemma \ref{contr_lemma_inf}. On the other hand, exploiting \eqref{eq:integral_j}, \eqref{eq:sici} and \eqref{eq-Q},
		\bml{
			\label{eq:a2}
  			A_2(\tau) = \frac{1}{\pi} q(0) \int_{\R^2}\dpv\:\frac{e^{-ip^2\tau}}{p^2+\lambda}  
             	= q(0) \lf[ -   e^{i\lambda\tau}(\gamma+\log\lambda+\log\tau)+\frac{e^{i\lambda\tau}}{\pi}Q(\lambda;\tau) \ri] \\[.2cm] 
              	=  q(0) \underbrace{e^{i\lambda\tau}\J(\tau)}_{A_{2,1}(\tau)}+ q_j(0) \underbrace{\tfrac{e^{i\lambda\tau}}{\pi}\left(-\pi\log\lambda+Q(\lambda;\tau)\right)}_{A_{2,2}(\tau)}.
		}
 		Now, it is clear that $A_{2,2}(\tau)$ is smooth, so that $IA_{2,2}$ is continuous. Furthermore, by \eqref{eq:i_identity},
 		\[
  			(IA_{2,1})(t)=1+\int_0^t\dtau\:\I(t-\tau) a_{2,1}(\tau), 	\qquad	 a_{2,1}(\tau) := \lf(e^{i\lambda\tau}-1 \ri)\J(\tau).
 		\]
		Since $ a_{2,1} $ is continuous (actually belongs to $H^1(0,T)$), then $IA_{2,1}$ is continuous too. Summing up, we have thus shown that $f$ (defined by \eqref{eq:mapf}) is continuous.
 
		For every $q\in \BB$, with $ \BB$ bounded,
		\[
 			{\lf|g(t,\tau,q)+\kappa\I(t-\tau)q \ri|\leq C\,\I(t-\tau)}
		\]
		and, since $\I\in L_{\mathrm{loc}}^1(\R^+)$, (ii) is satisfied.

		In addition, let $I=[a,b]$ be an interval, $\varphi:I\to\C$ a continuous function and $t_0\in\R^+$. The integral in \eqref{eq-lim_I} consists, up to some constants, of terms like
		\[
 			{\int_a^b\dtau\: \lf[\I(t-\tau)-\I(t_0-\tau) \ri]\,\lf[\beta_0|\varphi(\tau)|^{2\sigma}\varphi(\tau)+\kappa\varphi(\tau)\ri]}
		\]
		Hence (iii) is satisfied by dominated convergence (see, e.g., the discussion of \eqref{eq:dominated}). 
		
		Finally, we see that, as $q_1,\,q_2\in \BB$,
		\bmln{
 			{\displaystyle \lf|g(t,\tau,q_1)-g(t,\tau,q_2)+\kappa\I(t-\tau)(q_1-q_2) \ri| }\\ 
 	 		{\displaystyle \leq C\I(t-\tau)\left||q_1|^{2\sigma}q_1-|q_2|^{2\sigma}q_2\right|+|\I(t-\tau)||q_1-q_2|	 \leq C\,\I(t-\tau)|q_1-q_2|.}
}
		Consequently, setting $h(t,\tau)=C\,\I(t-\tau)$, (iv) is satisfied.
	\end{proof}

	\begin{pro}[Sobolev regularity of $ q(t) $]
		\label{pro:charge}
 		\mbox{}	\\
 		Let $\sigma\geq\frac{1}{2}$ and $  \psi_0\in \dom $. Then, there exists $ 0 < T < T_{*} $, such that $ q(t) \in H^{1/2}(0,T)$.
	\end{pro}
	
	\begin{proof} The key observation is that, if one proves that the map
		\beq
  			\mathcal{G}(q)[t]=f(t)-\int_0^t\dtau\:\bigg(g(t,\tau,q(\tau))+\kappa\I(t-\tau)q(\tau)\bigg)
 		\eeq
 is a contraction in a suitable subset of $C[0,T]\cap H^{1/2}(0,T)$, for a sufficiently small $T\in(0,T^*)$, then \eqref{eq:charge_compact} has a unique solution in this subset. Hence such a solution must coincide with the unique continuous solution provided by Proposition \ref{pro:charge continuous}, which thus belongs to $ H^{1/2}(0,T) $.
 
 For fixed $0<T < T^*$, the first point is to investigate the Sobolev regularity of the forcing term $f$. We know that $4\pi(U_0(\tau)\psi_0)(\yv)=A_1(\tau)+A_2(\tau)$, with $A_i$ defined in \eqref{eq:Aunoduetre}. Concerning $A_1$, we write
\bmln{
	A_1(\tau) = 2 \int_{\R^2} \dpv \: e^{i\pv \cdot \yv} e^{-i p^2\tau} \widehat{\phi_{\lambda,0}}(\pv) = 2 \int_{\mathbb{R}^{2}} \dpv \: e^{-i p^2\tau} \lf(\widehat{ T_{-\yv} \phi_{\lambda,0}} \ri)(\pv) \\	
  	= 2\pi \int_{0}^{\infty} \drho \: e^{-i \varrho \tau} \mean{\widehat{ T_{-\yv} \phi_{\lambda,0} }}(\sqrt{\varrho}) = (2\pi)^{3/2}  \widecheck{G_1}(-\tau)
}
 where $T_{\yv}$ is the translation operator, i.e., $ (T_{\yv} \psi)(\xv) : = \psi(\xv - \yv) $,
 \[
  	G_1(\varrho) : = \one_{[0,+\infty)}(\varrho) \: \mean{\widehat{ T_{-\yv} \phi_{\lambda,0} }}(\sqrt{\varrho}),
 \]
 and $ \mean{f} $ denotes the angular average of a function on $ \R^2 $, i.e.,
 \[
  \mean{f}(\varrho) = \frac{1}{2\pi} \int_0^{2\pi} \diff \vartheta \: f(\varrho\cos\vartheta, \varrho\sin\vartheta).
 \]
 Consequently, one finds that 
\bmln{
	\lf\| A_1 \ri\|_{H^{\nu}(\R)}^2 =  (2\pi)^3 \int_\R\drho\: \lf(1+\varrho^2 \ri)^\nu \lf|G_1(-\varrho) \ri|^2  = 
	 16 \pi^3 \int_{0}^{\infty} \dpi \: \lf(1+p^4 \ri)^{\nu} p \lf| \mean{\widehat{T_{-\yv}  \phi_{\lambda,0} }}(p) \right|^{2} \\                                                    
	 \leq C \int_{\R^2} \dpv \: \lf(1+p^4\ri)^{\nu}\lf|  \lf( \widehat{T_{-\yv}  \phi_{\lambda,0} } \ri) (\pv) \right|^{2}
}
so that $A_1\in H^{1/2}(0,T)$, since $\phi_{\lambda,0}\in H^1(\R^2)$ by assumption. As $A_1$ is bounded too we have, by Lemma \ref{contr_lemma}, that $IA_1\in H^{1/2}(0,T)$. On the other hand, since $A_{2,2}$ is smooth, $IA_{2,2}$ is smooth as well and, as $IA_{2,1}=1+I a_{2,1}$ with $a_{2,1}\in H^1(0,T)$, we have that $IA_{2,1}\in H^{1/2}(0,T)$. Summing up, recalling \eqref{eq:mapf} and \eqref{eq:Aunoduetre}, we proved that $f\in H^{1/2}(0,T)$.
 
We introduce now the contraction space: let
 \[
  	\A_{T} = \lf\{q\in C[0,T]\cap H^{1/2}(0,T) \: \big| \: \lf\|q\ri\|_{L^\infty(0,T)}+\lf\|q\ri\|_{H^{1/2}(0,T)}\leq \mathrm{b}_{T} \ri\},
 \]
 with $\mathrm{b}_{T}=2\max\{\|f\|_{L^\infty(0,T)}+\|f\|_{H^{1/2}(0,T)},1\}$. The set $ \A_{T}$ is a complete metric space with the norm induced by $C[0,T]\cap H^{1/2}(0,T)$, i.e.,
 \[
  	\lf\| \cdot \ri\|_{\A_{T}}= \lf\|\cdot \ri\|_{L^\infty(0,T)}+ \lf\|\cdot \ri\|_{H^{1/2}(0,T)}.
 \]
 In order to prove that $ \G $ defines a contraction on $ \A_{T} $, we need to show that $ \G $ maps $ \A_T $  into  itself and that the contraction condition on the norms is satisfied.

 We start by proving that $\mathcal{G}(\A_{T})\subset \A_{T}$. Letting $q\in \A_{T}$, one immediately sees that $\mathcal{G}(q)[t]$ is continuous. Then, we split the homogenous part of $\mathcal{G}(q)[t]$ into two terms:
 \[
  \mathcal{G}_1(q)[t]=\int_0^t\dtau\;g(t,\tau,q(\tau)),\qquad \mathcal{G}_2(q)[t]=\kappa\int_0^t\dtau\;\I(t-\tau)q(\tau).
 \]
 From \eqref{eq:lip_inf} and \eqref{eq:lip_sob}, \eqref{eq:integral_i}, \eqref{eq:est_I_lim}, one obtains
 \bdm
  	\lf\|\mathcal{G}_1(q) \ri\|_{H^{1/2}(0,T)} \leq C\lf\|I\,|q|^{2\sigma}q \ri\|_{\A_{T}}\leq C_T \lf\| |q|^{2\sigma}q \ri\|_{\A_{T}} \leq C_T \mathrm{b}_{T}^{2\sigma} \lf\| q \ri\|_{\A_T} \leq C_T  \mathrm{b}_{T}^{2\sigma+1},
\edm
where, from now on, $C_T$ stands for a generic positive constant such that $C_T \to 0$, as $ T \to 0 $, and which may vary from line to line. In addition, using \eqref{eq:lip_inf} and \eqref{eq:est_I_inf}, one sees that
\bdm
  	\lf\|\mathcal{G}_1(q)\ri\|_{L^\infty(0,T)} \leq C \lf\|I\,|q|^{2\sigma}q \ri\|_{L^\infty(0,T)} \leq C_T  \lf\||q|^{2\sigma}q \ri\|_{L^\infty(0,T)} \leq C_T\mathrm{b}_{T}^{2\sigma+1},
\edm
 so that
 \begin{equation}
  \label{eq:stimaguno}
  \lf\|\mathcal{G}_1(q)\ri\|_{\A_{T}}\leq C_T\mathrm{b}_{T}^{2\sigma+1}.
 \end{equation}
 On the other hand, we find that $ \lf\|\mathcal{G}_2(q) \ri\|_{H^{1/2}(0,T)}\leq C_T \lf\|q \ri\|_{\A_{T}}\leq C_T \mathrm{b}_{T}$, while, from \eqref{eq:est_I_inf}, $ \lf\| \mathcal{G}_2(q) \ri\|_{L^\infty(0,T)} \leq C_T\|q\|_{L^\infty(0,T)} \leq C_T\mathrm{b}_{T}$. Thus, we have
 \bdm
	\lf\|\mathcal{G}_2(q) \ri\|_{\A_{T}} \leq C_T \lf\| q \ri\|_{\A_{T}}\leq C_T \mathrm{b}_{T}.
\edm
Putting it together with \eqref{eq:stimaguno}, we finally get
\[
  \lf\|\mathcal{G}(q) \ri\|_{\A_{T}} \leq\mathrm{b}_{T} \left[  \frac{1}{2}+ C_T \left(1+\mathrm{b}_{T}^{2\sigma}\right)\right].
 \]
 Consequently, as the term in brackets is equal to $\frac{1}{2}+o(1)$ as $T\to0$, for $T$ sufficiently small $\mathcal{G}(q)\in \A_{T}$.
 
 Therefore, it is left to prove that $\mathcal{G}$ is actually a norm contraction. Given two functions $q_1,\,q_2\in \A_{T}$, we have
 \[
  \mathcal{G}(q_1)-\mathcal{G}(q_2)=\mathcal{G}_1(q_1)-\mathcal{G}_1(q_2)+\mathcal{G}_2(q_1-q_2).
 \]
 Arguing as before, one sees that $ \lf\|\mathcal{G}_2(q_1-q_2) \ri\|_{\A_{T}}\leq C_T\lf\|q_1-q_2 \ri\|_{\A_{T}}$. On the other hand, using again \eqref{eq:est_I_inf} and Lemma \ref{lem-lip} and \ref{contr_lemma},
 \bdm
  	\displaystyle \left\|I\left(|q_1|^{2\sigma}q_1-|q_2|^{2\sigma}q_2\right)\right\|_{\A_{T}} \leq C_T \left\||q_1|^{2\sigma}q_1-|q_2|^{2\sigma}q_2\right\|_{\A_{T}} \leq C_T \mathrm{b}_{T}^{2\sigma} \lf\|q_1-q_2 \ri\|_{\A_{T}}.
\edm
 Then,
 \[
  	\lf\|\mathcal{G}(q_1)-\mathcal{G}(q_2) \ri\|_{\A_{T}} \leq C_T \big(1+\mathrm{b}_{T}^{2\sigma}\big) \lf\|q_1-q_2 \ri\|_{\A_{T}}.
 \]
 Hence, since $C_T\to0$ as $T\to0$ and $\mathrm{b}_{T}$ is bounded, $\mathcal{G}$ is a contraction on $\A_{T}$, provided that $T$ is small enough.
\end{proof}

	The contraction time $ T $ provided by Proposition \ref{pro:charge} is a priori shorter than the maximal existence time of a continuous solution $ T_* $ given by Proposition \ref{pro:charge continuous}. However, we can extend the Sobolev regularity of $ q(t) $ up to $ T_{*} $. {In order to do that, however, two further auxiliary results are required. The first one concerns the log-H\"older regularity of the solution of the charge equation.}
	
	\begin{lem}
		\label{lemma:qlogcont}
		\mbox{}	\\
		Let $q(t)$ be the solution of \eqref{eq:charge_eq} provided by Proposition \ref{pro:charge} and $ T_* $ the maximal existence time given in Proposition \ref{pro:charge continuous}, then
		\[
			q(t) \in \clog[0,T],	\qquad		\forall \beta \leq 1,
		\]
		for any $ T < T_* $.
	\end{lem}
	
	\begin{proof}
		Fix $T<T_*$. We first remark that 
		the proof of $ f \in \clog[0,T] $ is equivalent to show that there exists $C>0$ for which
		\beq
			\label{eq:clogcondition}
			\lim_{\delta \to 0} |\log \delta|^{\beta} \lf| f(s+\delta) - f(s) \ri| \leq C < +\infty,
		\eeq
		for any $ s \in [0,T] $ (where at the extreme points the limit has to be suitably adjusted).
		
		From the charge equation \eqref{eq:charge_eq}, we get
		\bmln{
			q(t) = - 4\pi\beta_0\int_0^t\dtau\:\I(t-\tau)|q(\tau)|^{2\sigma}q(\tau)	\\
 			+2\left(\log 2-\gamma+i\tfrac{\pi}{4}\right)\int_0^t\dtau\:\I(t-\tau)q(\tau) + 4\pi\int_0^t\dtau\:\I(t-\tau)(U_0(\tau)\psi_0)(\yv),
		}
		i.e., $ q(t) = I_1(t) + I_2(t) + I_3(t) $ (with obvious meaning of the three terms).
		 
		Let us consider first $ I_1(t) $. The case $t=0$, $\delta>0$ is easier to deal with: since $q(t) $ is bounded on $[0,T]$,
		\[
			\lf| I_1(\delta)-I_1(0) \ri| \leq C \lf\| q \ri\|^{2\sigma + 1}_{L^\infty(0,T)}\int_0^\delta\dtau\:\I(\delta-\tau) \leq C \NN(\delta)\underset{\delta \to 0}{\sim}\frac{C}{|\log\delta|}
		\]
		where we recall the definition of $\NN$ given by \eqref{eq:N} and its asymptotic behavior in \eqref{eq:Nasympt}. On the other hand, if we consider the case $t\in(0,T)$, $\delta>0$ ($\delta<0$ is analogous), then we see that
		\bmln{
		 	I_1(t+\delta) - I_1(t) = C\int_t^{t+\delta} \diff \tau \: \I(t + \delta - \tau) |q(\tau)|^{2\sigma}q(\tau) \\
		 	+  C\int_0^t \diff \tau \: \lf[ \I(t + \delta - \tau) - \I(t- \tau) \ri] |q(\tau)|^{2\sigma_j}q(\tau) : = C\big(I_{1,1}(\delta,t) + I_{1,2}(\delta,t)\big).
		}
		Now, arguing as before, one easily finds that $I_{1,1}(\delta,t) \sim \frac{1}{|\log\delta|}$ as $ \delta \to 0 $ (independently of $t$). Furthermore, again by the boundedness of $ q $, one has
		\begin{equation}
		 	\label{eq:Ilog_aux}
		 	\lf| I_{1,2}(\delta,t) \ri| \leq C \int_0^t\dtau\: \lf|\I(\tau+\delta)-\I(\tau) \ri|.
		\end{equation}
		Since $\I$ is continuous, coercive and strictly convex \cite{CF,H}, it has a unique minimum point $t_{\mathrm{min}}>0$. If $t+\delta \leq t_{\mathrm{min}}$, then
		\[
		 	\int_0^t\dtau\:|\I(\tau+\delta)-\I(\tau)|=\NN(\delta)+\NN(t)-\NN(t+\delta)\underset{\delta \to 0}{\sim}\frac{1}{|\log\delta|}
		\]
		independently of $t$, by the log-H\"older continuity of $\NN$. Thus, combining with \eqref{eq:Ilog_aux}, one has $ I_{1,2}(\delta,t) \sim \frac{1}{|\log\delta|} $ as $  \delta \to 0 $. If, on the opposite, $t \geq t_{\mathrm{min}}$ (the case $ t < t_{\mathrm{min}} < t + \delta $ can be excluded for $ \delta $ small enough), then
		\bmln{
		 	\int_0^t\dtau\:|\I(\tau+\delta)-\I(\tau)| =  \int_0^{t_{\mathrm{min}}-\delta}\dtau\:\lf(\I(\tau) - \I(\tau+\delta) \ri) + \int_{t_{\mathrm{min}}-\delta}^{t_{\mathrm{min}}}\dtau\:  \lf|\I(\tau+\delta)-\I(\tau) \ri| \\
		 	+\int_{t_{\mathrm{min}}}^t \dtau\: \lf( \I(\tau+\delta)-\I(\tau) \ri) \leq  \NN(t_{\mathrm{min}}-\delta) - \NN(t_{\mathrm{min}}) + \NN(\delta) \\
		 	+ \NN(t + \delta) - \NN(t) - \NN(t_{\mathrm{min}}+\delta) + \NN(t_{\mathrm{min}}) + C \delta \leq  \NN(t + \delta) - \NN(t) + \NN(\delta) + C \delta
		}
		and, arguing as before, we obtain $I_{1,2}(\delta,t) \sim \frac{1}{|\log\delta|} $, as $ \delta  \to 0 $.
		
		Therefore, it is left to investigate the behavior of $I_2(t)$ and $I_3(t)$. First, one can easily see that $I_2(t)$ can be studied in the same way as $I_1(t)$. On the contrary, $I_3(t)$ requires some further remark, since $(U_0(\tau)\psi_0)(\yv)$ is not bounded on $[0,T]$. However, from \eqref{eq:Aunoduetre} it can be split into the sum of two terms  $A_1$ and $A_2$. The first one is bounded and hence it is possible to use the previous strategy to prove that $IA_1$ have the needed property. Concerning $A_2$, arguing as in the proof of Proposition \ref{pro:charge continuous}, one sees that it can be split, in turn, in two terms $A_{2,1}$ and $A_{2,2}$, where the second one is bounded and the first one satisfies the following property
		\beq
			\label{eq:log cancellation}
			\int_0^t\dtau\:\I(t-\tau)A_{2,1}(\tau)=1+\int_0^t\dtau\:\I(t-\tau)a_{2,1}(\tau)
		\eeq
		with $a_{2,1}(\tau)$ bounded. Consequently, $IA_2$ can be bounded exactly as above.
	\end{proof}
	
	{The second auxiliary result is a slight modification of Lemma \ref{contr_lemma_inf} and Lemma \ref{contr_lemma}.}
	
	{\begin{lem}
	\label{lem:contr_further}
	\mbox{}	\\
	Let $T>0$ and $h\in C[0,T]\cap H^{1/2}(0,T)$. Then
	$$
	 	\h(t):=\int_0^T\dtau\:\I(t+T-\tau)h(\tau)
	$$
	 belongs to $C[0,\T]\cap H^{1/2}(0,\T)$ for any $\T>0$.
	\end{lem}}
	
	{\begin{rem}
	 The key point in Lemma above is that, unlike in Lemma \ref{contr_lemma_inf} and Lemma \ref{contr_lemma}, the integration kernel do not present singularities, being shifted of $T>0$. In addition, in Lemma \ref{lem:contr_further} only the preservation of the regularity is investigated and not the contractive properties.
	\end{rem}}

	\begin{proof}
	 One can easily see that $\h\in L^2(0,\T)$ for all $\T>0$. In addition, simply repeating the argument of Lemma \ref{contr_lemma_inf}, one finds that $\h\in C[0,\T]$.
	 
	 Then, it is left to prove that
	 \begin{equation}
	 \label{eq:newt}
	 \tfrac{1}{2}\big[\h\big]_{\dot{H}^{1/2}(0,\T)}^2 =\int_0^{\T}\dt\int_0^t\ds\:\frac{\big|\h(t)-\h(s)\big|^2}{|t-s|^2}<\infty.
	 \end{equation}
	Preliminarily, we note that $\h$ can be rewritten as
	 \[
	 \h(t)=\int_t^{t+T}\dtau\:\I(\tau)h(t+T-\tau)
	 \]
	and hence that, when $s\leq t\leq s+T$,
	\bml{
	  \label{eq:division}
	 \h(t)-\h(s)= -\int_s^t\dtau\:\I(\tau)h(s+T-\tau)+\int_{s+T}^{t+T}\dtau\:\I(\tau)h(t+T-\tau)	\\
	                      +\int_t^{s+T}\dtau\I(\tau)\big[h(t+T-\tau)-h(s+T-\tau)\big].
	 }
	 
	Assume, first, that $\T\leq T$ (whence $t\leq s+T$). By \eqref{eq:division}
	 \bmln{
	 \tfrac{1}{2}[\h]_{\dot{H}^{1/2}(0,\T)}^2\leq C\bigg[ \underbrace{ \int_0^{\T}\dt\int_0^t\ds\:\frac{1}{|t-s|^2}\bigg|\int_s^t\dtau\:\I(\tau)h(s+T-\tau)\bigg|^2}_{:=B_1}	\\
	                                                        + \underbrace{\int_0^{\T}\dt\int_0^t\ds\:\frac{1}{|t-s|^2}\bigg|\int_{s+T}^{t+T}\dtau\:\I(\tau)h(t+T-\tau)\bigg|^2}_{:=B_2}	\\
	                                                      +\underbrace{ \int_0^{\T}\dt\int_0^t\ds\:\frac{1}{|t-s|^2}\bigg|\int_{t}^{t+T}\dtau\:\I(\tau)\big[h(t+T-\tau)-h(s+T-\tau)\big]\bigg|^2}_{:=B_3}\bigg].
	 }
	 Now, one can easily see that
	 \[
	  B_1+B_2\leq\|h\|_{L^\infty(0,T)}^2 \lf( [\NN(\: \cdot \; ) ]_{\dot{H}^{1/2}(0,\T)}^2 + [\NN( \: \cdot + T)]_{\dot{H}^{1/2}(0,\T)}^2 \ri)<\infty.
	 \]
	 On the other hand, by Jensen inequality and monotonicity of $\NN$,
	 \[
	  B_3\leq\NN(T+\T)\int_0^{\T}\dt\int_0^t\ds\int_t^{s+T}\dtau\:\I(\tau)\frac{|h(t+T-\tau)-h(s+T-\tau)|^2}{|t-s|^2}.
	 \]
	Consequently, splitting the integral in $\tau$ in the two domains $[t,T]$ and $[T,T+s]$ and using twice (for each of the two terms) Fubini theorem, e.g., as in
	 \bdm
	 	\int_0^t \diff s \int_T^{s+T} \diff \tau \: F(s,\tau) = \int_T^{t+T} \diff \tau \int_0^{\tau - T} \diff s \:  F(s,\tau),
	\edm
	 one obtains that
	 \[
	 B_3\leq2\NN^2(2T)[h]_{\dot{H}^{1/2}(0,T)}^2<\infty.
	 \]
	Thus we proved that, if $\T\leq T$, then $\h\in H^{1/2}(0,\T)$.
	 
	 Consider the opposite case when $\T>T$. Here
	 \[
	  \tfrac{1}{2}[\h]_{\dot{H}^{1/2}(0,\T)}^2=\int_0^T\dt\int_0^t\ds\:\frac{|\h(t)-\h(s)|^2}{|t-s|^2}+\int_T^{\T}\dt\int_0^t\ds\:\frac{|\h(t)-\h(s)|^2}{|t-s|^2}.
	 \]
	 First, one observes that the first term above can be managed exactly as in the case $\T\leq T$ and therefore is finite. The second one on the other hand can be split as
	 \[
	  \underbrace{\int_T^{\T}\dt\int_{t-T}^t\ds\:\frac{|\h(t)-\h(s)|^2}{|t-s|^2}}_{:=D_1}+\underbrace{\int_T^{\T}\dt\int_0^{t-T}\ds\:\frac{|\h(t)-\h(s)|^2}{|t-s|^2}}_{:=D_2}.
	 \]
	 The estimate of $D_2$ is immediate since $t-s \geq T$, so that
	 \[
	 D_2\leq\frac{1}{T^2}\int_T^{\T}\dt\int_0^{t-T}\ds\:|\h(t)-\h(s)|^2<\infty.
	 \]
	On the other hand, recalling \eqref{eq:division},
	 \bmln{
	  D_1\leq C\bigg[\underbrace{\int_T^{\T}\dt\int_{t-T}^t\ds\:\frac{1}{|t-s|^2}\bigg|\int_s^t\dtau\:\I(\tau)h(s+T-\tau)\bigg|^2}_{:=D_{1,1}}	\\
	                  + \underbrace{\int_T^{\T}\dt\int_{t-T}^t\ds\:\frac{1}{|t-s|^2}\bigg|\int_{s+T}^{t+T}\dtau\:\I(\tau)h(t+T-\tau)\bigg|^2}_{:=D_{1,2}}	\\
	                  +\underbrace{\int_T^{\T}\dt\int_{t-T}^t\ds\:\frac{1}{|t-s|^2}\bigg|\int_{t}^{t+T}\dtau\:\I(\tau)\big[h(t+T-\tau)-h(s+T-\tau)\big]\bigg|^2}_{:=D_{1,3}}\bigg].
	 }
	Now, arguing as in the first part of the proof, namely exploiting the regularity of $\NN$, the assumptions on $h$, Jensen inequality and Fubini theorem, one can check that
	 \[
	 D_{1,1}+D_{1,2}\leq \|h\|_{L^\infty(0,T)}^2[\NN]_{\dot{H}^{1/2}(0,T+\T)}^2<\infty
	 \]
	 and
	 \[
	  D_{1,3}\leq C\NN^2(T+\T)[h]_{\dot{H}^{1/2}(0,T)}^2<\infty,
	 \]
	 which, then, concludes the proof.
	\end{proof}
	
	\begin{pro}[Regularity extension of $ q(t) $]
		\label{pro:extension q}
		\mbox{}	\\
		Let $q(t)$ be the solution of \eqref{eq:charge_eq} provided by Proposition \ref{pro:charge} and $ T_* $ the maximal existence time given in Proposition \ref{pro:charge continuous}, 
		then $ q(t) \in H^{1/2}(0,T) $ for any $ T < T_* $.
	\end{pro}
		
	\begin{proof}
		{Let $q$ be the solution of \eqref{eq:charge_eq} provided by Proposition \ref{pro:charge continuous}. First, by Proposition \ref{pro:charge}, there exists $T_1\in(0,T_*)$ such that $q\in H^{1/2}(0,T_1)$.}
	
	{Now consider the equation
	\begin{equation}
	 \label{eq:charge_new}
	q_1(t) + \int_0^t \diff \tau \: \bigg(g(t,\tau,q_1(\tau))+\kappa\I(t-\tau) \: q_1(\tau)\bigg) = f_1(t),
	\end{equation}
	where
	\[
	 f_1(t):=f(t+T_1)-4\pi\beta_0\int_0^{T_1}\dtau\:\I(t+T_1-\tau)|q (\tau)|^{2\sigma}q (\tau)-\kappa\int_0^{T_1}\dtau\:\I(t+T_1-\tau)q (\tau).
	\]
	From the regularity properties of $f$ established in Propositions \ref{pro:charge continuous} and \ref{pro:charge} and exploiting Lemma \ref{lem:contr_further} with $T=T_1$ and $h=4\pi\beta_0|q|^{2\sigma}q+\kappa q$, one can see that $f_1\in C[0,T]\cap H^{1/2}(0,T)$ for every $T<T_*-T_1$ (recall that $|q|^{2\sigma}q\in C[0,T_1]\cap H^{1/2}(0,T_1)$ by Lemma \ref{lem-lip}). Consequently, arguing as in the proofs of Propositions \ref{pro:charge continuous} and \ref{pro:charge}, there exists $T_1'<T_*-T_1$ and $q_1\in C[0,T_1']\cap H^{1/2}(0,T_1')$ which solves \eqref{eq:charge_new}. In addition, an easy computations shows that $q(t)=q_1(t-T_1)$ for every $t\in[T_1,T_1+T_1']$, so that we have found a solution to the charge equation such that $q\in H^{1/2}(0,T_1)$ and $q\in H^{1/2}(T_1,T_1+T_1')$. In principle, this would not be sufficient in order to claim that $q\in H^{1/2}(0,T_1+T_1')$ due to the non-locality of the $H^{1/2}$-norm and the failure of the Hardy inequality (as explained before Proposition \ref{pro:extension}). However, thanks to Lemma \ref{lemma:qlogcont}, we know a priori that $q\in C_{\log,1}[0,T]$ for all $T<T_*$, so that one can argue as in the proof of Proposition \ref{pro:extension} and obtain that actually $q\in H^{1/2}(0,T_1+T_1')$.}
	
	{This shows that once the regularity is proven up to a time $ T_1 \in (0, T_{*}) $, then it can be extended up to $ T_1 + T_1' < T_{*} $. A priori this procedure could stop before $ T_* $, e.g., because $ T_1' $ tends to $ 0 $ as we get closer and closet to $ T_{*}$. We can prove however that this is not the case.	Define $\widehat{T}:\sup\{T>0:q\in H^{1/2}(0,T)\}$, which is strictly positive by Proposition \ref{pro:charge}. In order to conclude, we must prove that $\widehat{T}=T_*$. Assume, then, by contradiction that $\widehat{T}<T_*$. Consequently, $q\in H^{1/2}(0,T)$ for every $T<\widehat{T}$ and $\|q\|_{L^\infty(0,\widehat{T})} < +\infty$. In addition, fix $\ep>0$ such that $\NN(\widehat{T}-T_\ep)\big(\|q\|_{L^\infty(0,\widehat{T})}^{2\sigma}+1\big)<1/2C$, where $T_\ep:=\widehat{T}-\ep$ and $C$ is a fixed constant that will be specified in the following, and $0<\delta<\ep$, so that $T_\delta:=\widehat{T}-\delta\in(T_\ep,\widehat{T})$.}
	
	{At this point we can estimate $\|q\|_{H^{1/2}(T_\ep,T_\delta)}$ by using \eqref{eq:charge_compact}. First we note that (let $h$ be defined as before) for $t\in(T_\ep,T_\delta)$
	\[
	 q(t)=f(t)-\int_0^{T_\ep}\dtau\:\I(t-\tau)h(\tau)-\int_{T_\ep}^t\dtau\:\I(t-\tau)h(\tau).
	\]
	Since $f\in H^{1/2}(0,T)$ for every $T>0$ (see the proof of Proposition \ref{pro:charge}), its $H^{1/2}(T_\ep,T_\delta)$-norm can be easily estimated independently of $\delta$. The same can be proved for the second term, arguing as in the proof of Lemma \ref{lem:contr_further} and noting that $\I(t-\tau)=\I(t'+T_\ep-\tau)$ with $t'\in[0,T_\delta-T_\ep]$. Summing up, 
	\begin{equation}
	 \label{eq:auxxxx}
	 \lf\|q\ri\|_{H^{1/2}(T_\ep,T_\delta)}\leq C_{\widehat{T},T_\ep}+\bigg\|\int_{T_\ep}^{(\cdot)}\dtau\:\I(\cdot-\tau)h(\tau)\bigg\|_{H^{1/2}(T_\ep,T_\delta)}
	\end{equation}
	(precisely, $C_{\widehat{T},T_\ep}$ depends only on $\|q\|_{L^\infty(0,\widehat{T})}$ and $\|q\|_{H^{1/2}(0,T_\ep)}$, which are finite quantities). Therefore, we have to estimate the last term on the r.h.s.. Since the $L^2$ norm can be easily estimated independently of $\delta$, it suffices to consider the seminorm of such a term: for any $T_\ep<s<t<T_\delta$, 
	 \bmln{
	  \displaystyle \int_{T_\ep}^t\dtau\:\I(t-\tau)h(\tau)-\int_{T_\ep}^s\dtau\:\I(s-\tau)h(\tau)  \\
	  \displaystyle =\int_{s-T_\ep}^{t-T_\ep}\dtau\:\I(\tau)h(t-\tau)+\int_0^{s-T_\ep}\dtau\:I(\tau)[h(t-\tau)-h(s-\tau)]
	}
	and hence, arguing exactly as in the proof of Proposition \ref{contr_lemma}, we can find that
	\[
	\left[\int_{T_\ep}^{(\cdot)}\dtau\:\I(\cdot-\tau)h(\tau)\right]_{\dot{H}^{1/2}(T_\ep,T_\delta)}\leq C \lf(\|h\|_{L^\infty(0,\widehat{T})}[\NN]_{\dot{H}^{1/2}(0,\widehat{T}-T_\ep)}+[h]_{\dot{H}^{1/2}(T_\ep,T_\delta)}\NN(\widehat{T}-T_\ep) \ri).
	\]
	Now, we can note that
	\[
	\|h\|_{L^\infty(0,\widehat{T})}\leq C\big(\|q\|_{L^\infty(0,\widehat{T})}^{2\sigma}+1\big)\|q\|_{L^\infty(0,\widehat{T})}
	\]
	and, using \eqref{eq:help}, that
	\[
	[h]_{\dot{H}^{1/2}(T_\ep,T_\delta)}\leq C\big(\|q\|_{L^\infty(0,\widehat{T})}^{2\sigma}+1\big)\|q\|_{H^{1/2}(T_\ep,T_\delta)}.
	\]
	Consequently, recalling \eqref{eq:auxxxx} and the definition of $\ep$ (and possibly redefining $C_{\widehat{T},T_\ep}$)
	\[
	 \lf\|q\ri\|_{H^{1/2}(T_\ep,T_\delta)}\leq C_{\widehat{T},T_\ep}+C\NN(T-T_\ep)\big(\|q\|_{L^\infty(0,\widehat{T})}^{2\sigma}+1\big)\lf\|q\ri\|_{H^{1/2}(T_\ep,T_\delta)}\leq C_{\widehat{T},T_\ep}+ \tx\frac{1}{2} \lf\|q\ri\|_{H^{1/2}(T_\ep,T_\delta)}.
	\]
	Hence, moving the last term to the l.h.s., we see that $\lf\|q\ri\|_{H^{1/2}(T_\ep,T_\delta)}$ can be estimated independently of $\delta$ and thus, letting $\delta \to 0$, there results that $\lf\|q\ri\|_{H^{1/2}(T_\ep,\widehat{T})}<\infty$. Summing up, we have that $q\in H^{1/2}(0,T_\ep)$ and $q\in H^{1/2}(T_\ep,\widehat{T})$ and, by log-H\"older continuity this means that $q\in H^{1/2}(0,\widehat{T})$. However, now one can use the first part of the proof with $T_1=\widehat{T}$ thus proving that there exists the possibility of a contraction argument beyond $\widehat{T}$, but this contradicts the definition of $\widehat{T}$, so that we proved that $\widehat{T}=T_*$.}
	\end{proof}

We can now prove Theorem \ref{teo:local}, since the existence and uniqueness of the charge $ q(t) $ will imply that the ansatz \eqref{eq:ansatz} is a solution to the weak Cauchy problem \eqref{eq:cauchyweak}. In order to see that and make the derivation of the charge equation discussed in Sect. \ref{sec:charge} correct, we need to handle integral expressions involving the derivative of $ q(t) $. This will be done as explained in the following Remark.

	\begin{rem}[Integration of $ \dot{q} $ -- part I]
 		\label{rem:meaning}
 		\mbox{}	\\
 In the following of the paper we will often manage integrals involving the distributional derivative of the charge $q(t)$. Clearly, such a notation is purely formal since we do not actually know whether $ q(t) $ is an absolutely continuous function. Hence, its derivatives might not be integrable in Lebesgue sense. However, $ q(t) \one_{[0,T]} $ is a compactly supported distribution belonging to $ \mathcal{E}^{\prime} $, the dual of $ \mathcal{E} = C^{\infty}(\R) $. Hence, its distributional derivative is well defined and it still belongs to $ \mathcal{E}^{\prime} $. On the other hand, for any continuous function $ f $, one obviously has
		\[
			{\dot{f}(t) \one_{[0,T]} = \frac{\diff}{\diff t} \lf( f(t) \one_{[0,T]} \ri) + f(T) \delta(t - T) - f(0) \delta(t)},
		\]
		and since the r.h.s. is in $ \mathcal{E}^{\prime} $, the same holds for the l.h.s.. Hence we can give a meaning to the expression
		\[
			\int_{0}^t \diff \tau \: g(\tau) \dot{q}(\tau), 
		\]		
		whenever $ g \in C^{\infty}(\R) $, as the distributional pairing between $ \mathcal{E}^{\prime} $ and $ \mathcal{E} $. Of course the above is not the Lebesgue integral and we should have used a different symbol. However, in order to avoid a too heavy notation, we make a little abuse and keep the same integral symbol. Note that with such a convention we actually have
 		\[
			\int_{0}^t \diff \tau \: \dot{q}(\tau) = q(t) - q(0),
		\]
		since the function $ 1 $ is smooth. 
		
		Of course if we knew a priori that $ q \in W^{1,1}(0,T) $, then there would be no problem in the definition of any integral involving $ \dot{q} $ against a continuous function.
	\end{rem}

	\begin{proof}[Proof of Theorem \ref{teo:local}] 
		Let $\psi_t$ be the function defined by \eqref{eq:ansatz} and \eqref{eq:charge_eq}. For the sake of simplicity we split the proof in two steps. In the former we show that $\psi_t\in\dom[\F]$, in the latter,  we prove that $ \psi_t $ is a solution of the weak problem \eqref{eq:cauchyweak}.

		In order to prove that $ \psi_t \in \dom[\F] $, it is sufficient to show that
		\begin{equation}
 			\label{eq:regHuno}
 			\psi_t(\xv) - \frac{1}{2\pi}q(t)K_0 \lf(\sqrt{\lambda}|\xv-\yv|\ri)\in H^1(\R^2).
		\end{equation}

		Exploiting \eqref{eq:ansatz} and the Fourier transform, we can see that the previous expression reads
\[
 e^{-ip^2t}\widehat{\psi_0}(\pv)+\frac{i}{2\pi}\int_0^t\dtau\:e^{-i\pv\cdot\yv}e^{-ip^2(t-\tau)}q(\tau)-\frac{1}{2\pi}\frac{q(t)e^{-i\pv\cdot\yv}}{p^2+\lambda}.
\]
Hence, integrating by parts (in view of Remark \ref{rem:meaning}), one finds
\begin{equation}
 \label{eq:regFourier}
 e^{-ip^2t}\left(\widehat{\psi_0}(\pv)-\frac{q(0)e^{-i\pv\cdot\yv}}{2\pi(p^2+\lambda)}\right)-\frac{e^{-i\pv\cdot\yv}}{2\pi(p^2+\lambda)}\int_0^t\dtau\:e^{-ip^2(t-\tau)}(\dot{q}(\tau)-i\lambda q(\tau)).
\end{equation}
Note that the integral of $ \dot q $ on the r.h.s. has to be understood as explained in Remark \ref{rem:meaning}, which can be done since $ e^{-ip^2(t - \tau)} $ is a smooth function of $ \tau $.

Now, if this function belongs to $L^2(\R^2,(p^2+1)\dpv)$, then \eqref{eq:regHuno} is fulfilled. For the first term this is immediate since it represents the Fourier transform of $U_0(t)\phi_{\lambda,0}$, which is in $H^1(\R^2)$, since $ \phi_{\lambda,0} $ does. Concerning the second term, we first set $ \lambda = 1 $ for the sake of simplicity, and change variables to get
\bmln{
  \displaystyle \int_{\R^2}\dpv\:(1+p^2)\bigg|\frac{e^{-i\pv\cdot\yv}}{2\pi(p^2+1)}\int_0^t\dtau\:e^{-ip^2(t-\tau)}(\dot{q}(\tau)-i q(\tau))\bigg|^2	\\
  \displaystyle \leq C\int_0^\infty\drho\:\frac{1}{1+\varrho}\bigg[\bigg|\int_0^t\dtau\:e^{i\varrho\tau}\dot{q}(\tau)\bigg|^2+\bigg|\int_0^t\dtau\:e^{i\varrho\tau}q(\tau)\bigg|^2\bigg].
}
Now, one can check that
\[
 \int_0^t\dtau\:e^{i\varrho\tau}\dot{q}(\tau) = \sqrt{2\pi} \: \widehat{{\dot \xi}} (-\varrho),		\qquad	\int_0^t\dtau\:e^{i\varrho\tau}q(\tau)=  \sqrt{2\pi} \: \widehat{\one_{[0,t]}q}(-\varrho),
\]
where
\beq
	\label{eq:xi}
	\xi(\tau):=
 	\begin{cases}
 		q(0),    & \mbox{if }\tau\leq0,\\	
  		q(\tau), & \mbox{if }0<\tau<t,\\	
 		q(t),    & \mbox{if }\tau\geq t.
 \end{cases}
\eeq
Note that $ \dot{\xi}  $ is a distribution belonging to $ \mathcal{E}^{\prime}$, as $ \dot{q} \one_{[0,t]} $ does, therefore we can define its Fourier transform, which is in fact a smooth function {\cite[Theorem 7.1.14]{Ho}}. Consequently,
\[
 \int_0^\infty\drho\:\frac{1}{1+\varrho}\bigg[\bigg|\int_0^t\dtau\:e^{i\varrho\tau}\dot{q}(\tau)\bigg|^2+\bigg|\int_0^t\dtau\:e^{i\varrho\tau}q(\tau)\bigg|^2\bigg]\leq  2\pi \int_\R\drho\:\frac{ | \widehat{{\dot \xi}} (\varrho) |^2}{1+|\varrho|}+\int_\R\drho\:\frac{|\widehat{\one_{[0,t]}q}(\varrho)|^2}{1+|\varrho|}.
  \]
{Now, as $q\in C_{\log,1}[0,T]\cap H^{1/2}(0,T)$ (arguing as in the proof of Proposition \ref{pro:extension}) one can see that $\xi\in H_{\mathrm{loc}}^{1/2}(\R)$. Consequently, $\dot{\xi}\in H_{\mathrm{loc}}^{-1/2}(\R)$ and, recalling that it is also compactly supported (on $[0,t]$), this entails that $\dot{\xi}\in H^{-1/2}$. Hence,} the r.h.s. of the previous inequality is finite. Summing up, \eqref{eq:regFourier} belongs to $L^2(\R^2,(p^2+1)\dpv)$ and then \eqref{eq:regHuno} is satisfied.

{The fact that the ansatz \eqref{eq:ansatz_pre} provides the unique solution to the weak problem \eqref{eq:cauchy_lin_form}, if $ q(t) $ solves the charge equation, is proven in detail in a slightly different setting (linear moving point interactions) \cite[Theorem 2.1]{DFT2} (see also \cite{CCF}) and we omit this discussion here for the sake of brevity.}
\end{proof}


\subsection{Conservation laws}
\label{sec:conservation}

In this section we prove the conservation of mass and energy claimed in Theorem \ref{teo:conservation}, which in turn will be the key to prove the global existence stated in Theorem \ref{teo:global}. We recall the results proven in Propositions \ref{pro:charge continuous}, \ref{pro:charge} and \ref{pro:extension q}: there exists some $ T_* > 0  $ such that there is a unique continuous solution of \eqref{eq:charge_eq} in $[0,T_*)$, which also belongs to $ H^{1/2}(0,T) $ for any $ 0 < T < T_* $. 

Before proceeding further, however, another Remark is in order about the use we will make of the derivative of $ q $. In view of Remark \ref{rem:meaning} it can be ``integrated'' against smooth functions by exploiting the distributional pairing. Here we aim at giving a meaning to some more singular expressions:

	\begin{rem}[Integration of $ \dot q(t) $ -- part II]
		\label{rem:meaning2}
		\mbox{}	\\
 		Thanks to Proposition \ref{pro:charge}, we know that as $ T < T_* $, $ q \in H^{1/2}(0,T) $. We claim that this is sufficient to give a rigorous meaning to the expression
 		\bdm
 			\int_0^T \diff t \: f(t) \dot q(t),
		\edm
		for any function $ f \in \clog[0,T] \cap H^{1/2}(0,T) $, $ \beta > 1/2 $. The idea is to use the pairing provided by the duality between $ H^{1/2}(\R) $ and $ H^{-1/2}(\R) $, which allows to interpret the integral of $ f^* g $, with $ f \in  H^{1/2}(\R) $ and $ g \in H^{-1/2}(\R) $, as
		\beq
			\label{eq:hpairing}
			\int_{\R} \diff t \: f^*(t) g(t) = \int_{\R} \diff p \: \lf( \sqrt{p^2 + 1} \widehat{f}^*(p) \ri) \lf( \tx\frac{1}{\sqrt{p^2 +1}} \widehat{g}(p) \ri),
		\eeq
		where the symbol on the l.h.s. is not the Lebesgue integral, while on the r.h.s. we are integrating the product of two $ L^2 $ functions. Note that such a duality fails in general on a compact subset of the real line.
		
		So, if $f \in \clog[0,T] \cap H^{1/2}(0,T) $, we can rewrite
		\bdm
			\int_{0}^T \diff t \: f(t) \dot{q}(t) = f(T) \lf( q(T) - q(0) \ri) + \int_{0}^T \diff t \: \lf( f(t) - f(T) \ri) \dot{q}(\tau)
		\edm
		and, since both $ f $ and $ q $ are continuous, $ f(T) $ is well defined as well as $ q(T) $ and $ q(0) $. Next we observe that $ f(t) - f(T) $ satisfies the hypothesis of Proposition \ref{pro:extension} with $ \beta > 1/2 $ and therefore there exists an extension $ \fe \in H^{1/2}(\R) $ of $ f(t) - f(T) $, such that
		\bdm
			\int_{0}^T \diff t \: \lf( f(t) - f(T) \ri) \dot{q}(\tau) = \int_{\R} \diff t \:  \fe(t)  \dot{\xi}(\tau),
		\edm
		where $ \xi $ is defined in \eqref{eq:xi}. Here we are using that $ \mathrm{supp}(\dot\xi) \subset  [0,T] $.
		Now, since $  f_{\mathrm{e}} \in H^{1/2}(\R) $ and $ \dot \xi \in H^{-1/2}(\R) $ (see \cite{CCF} and the proof of Theorem \ref{teo:local}), then the last integral is meant as in \eqref{eq:hpairing}.
	\end{rem}

Before attacking the proof, we state a technical Lemma, which is a consequence of the charge equation \eqref{eq:charge_eq} and which will be used in the derivation of the mass and energy conservation.

	\begin{lem}
 		\label{lem-inv}
		\mbox{}	\\
 		Let $q(t)$ be the solution of \eqref{eq:charge_eq} provided by Proposition \ref{pro:charge} and $ T_* $ the maximal existence time given in Proposition \ref{pro:charge continuous}, then 
 		\begin{equation}
  			\label{eq:charge_inv}
  			\lf(U_0(t)\psi_0\ri)(\yv) = \left(\beta_0|q(t)|^{2\sigma}+\frac{\gamma-\log 2}{2\pi}\right)q(t)-\frac{iq(t)}{8}+\frac{1}{4\pi}\frac{\diff}{\diff t}\int_0^t\dtau\:(-\gamma-\log(t-\tau))q(\tau),
		\end{equation}
 		for a.e. $t\in[0,T]$ with $ T < T_* $.
	\end{lem}

	\begin{proof}
 		Dividing the charge equation \eqref{eq:charge_eq} by $4\pi$, applying the operator $J$ defined by \eqref{eq:integral_j} and recalling Lemma \ref{lem:i_identity}, we obtain
 		\[
 		 \frac{1}{4\pi}(Jq)(t)+\int_0^t\dtau\:\lf(\beta_0|q(\tau)|^{2\sigma}+\frac{\gamma-\log 2}{2\pi}-\frac{i}{8} \ri)q(\tau)= \int_0^t\dtau\:(U_0(\tau)\psi_0)(\yv)
 		\]
 		and, in particular, that $Jq$ is absolutely continuous. Then, differentiating in $t$ and rearranging terms, we obtain \eqref{eq:charge_inv}.
	\end{proof}
	
	In view of Remark \ref{rem:meaning2} the following technical results will prove to be very useful.

	\begin{lem}
		\label{lemma:forzcont}
		\mbox{}	\\
		{Let $\phi$ be a function in $H^1(\R^2)$ such that $(1+p)^\ep\widehat{\phi}\in L^1(\R^2)$, for some $\ep>0$, and let $ \mathcal{J}(t) $ be the function defined in \eqref{eq:integral_j}. Then, for every $\xv,\yv\in\R^2$, $\xv\neq\yv$, and every $T\in\R^+$,}
		 {\beqn
			\label{eq:forzcont1}
			\lf( U_0(\cdot) \phi \ri)(\xv) & \in & \clog[0,T] \cap H^{1/2}(0,T),			\\
		\label{eq:forzcont3}
			\lf( U_0(\cdot) K_0(|\cdot \: - \yv| \ri))({\yv})-\tfrac{e^{it}}{2}\J(\cdot) & \in & \clog[0,T] \cap H^{1/2}(0,T),		
		\eeqn
		for any $ \beta \in \R^+ $.}
	\end{lem}
	
	{\begin{rem}
	 Note that if a function $\psi\in\dom$ (see \eqref{eq:in_assumption}), then its regular part $\phi_{1,0}$ satisfies \eqref{eq:forzcont1}, whereas its singular part (choosing $\lambda=1$ for the sake of simplicity) undergoes 
	 \eqref{eq:forzcont3}.
	\end{rem}}

	\begin{proof}[Proof of Lemma \ref{lemma:forzcont}]
		{In order to prove that $ \lf( U_0(t) \phi\ri)(\xv) \in  H^{1/2}(0,T) $, for any finite $ T > 0 $, one can argue as in the proof of Proposition \ref{pro:charge} (simply replacing $\phi_{1,0}$ with $\phi$).} 
		
		Hence, we have only to verify the other properties. By expressing the quantity using the Fourier transform, we have
		\bdm
			\lf| \lf( U_0(t+\delta) \phi \ri)(\xv) - \lf( U_0(t) \phi \ri)(\xv) \ri| \leq \int_{\R^2} \diff \pv \: \lf| e^{-i p^2 \delta} - 1 \ri| \lf| \widehat{\phi}(\pv) \ri|.
		\edm
		Again, in order to show that $ \lf( U_0(t) \phi_{1,0} \ri)(\xv) \in \clog[0,T] $, it suffices to prove that the analogue of \eqref{eq:clogcondition} holds true. To this aim we observe that, for any $ \eps > 0 $,
		\bdm
			|\log\delta|^{\beta} \lf| e^{-i p^2 \delta} - 1 \ri| \leq 2^{1- \eps/2} p^{\eps} \delta^{\eps/2} |\log\delta|^{\beta} \xrightarrow[\delta \to 0]{} 0,\qquad\forall\beta \in\R^+,
		\edm
		{and thus the result is a} direct consequence of the properties of {$\phi$} and dominated convergence.		
		
		{On the other hand, 
		the regularity of $\lf( U_0(\cdot) K_0(|\cdot \: - {\yv}| \ri))({\yv})$ can be proved simply using \eqref{eq:a2} and the smoothness of $Q(1;t)$ (defined by \eqref{eq-Q}).}
	\end{proof}

	\begin{lem}
		\label{lemma:dotq}
		\mbox{}	\\
 		Let $q(t)$ be the solution of \eqref{eq:charge_eq} provided by Proposition \ref{pro:charge} and $ T_* $ the maximal existence time given in Proposition \ref{pro:charge continuous}, then 
 		\[
 			\int_0^{t} \diff \tau \: \lf(- \gamma - \log(\cdot - \tau) \ri) \dot q(\tau) {: = (\gamma + \log t) q(0) + \frac{\diff}{\diff t} \int_0^t \diff \tau \:\ \lf(- \gamma - \log(t - \tau) \ri) q(\tau) =: B_{q}(t)} 
		\]
		{belongs to $ \clog[0,T] \cap H^{1/2}(0,T) $} for every $\beta\in(0,1]$ and every $0< T < T_* $.
	\end{lem}
	
	\begin{proof}
		{From \eqref{eq:charge_inv}, we have that 
		\bmln{
		 \frac{\diff}{\diff t} \int_0^t\dtau\:(-\gamma-\log(t-\tau))q(\tau) =4\pi(U_0(t)\psi_0)(\yv)-4\pi\left(\beta_0|q(t)|^{2\sigma}+\frac{\gamma-\log 2}{2\pi}\right)q(t)	\\
		 +\frac{i\pi q(t)}{2}.
		}
		Now, the weak derivative of the l.h.s. reads
		\bdm
			\frac{\diff}{\diff t} \int_0^t \diff \tau \:\ \lf(- \gamma - \log(t - \tau) \ri) q(\tau) = -(\gamma + \log t) q(0) + \int_0^t \diff \tau \: \lf(- \gamma - \log(t - \tau) \ri) \dot q(\tau),
		\edm
		so that
		\bml{
		 \label{eq-aux}
		 \int_0^t \diff \tau \: \lf(- \gamma - \log(t - \tau) \ri) \dot q(\tau) =  4\pi(U_0(t)\psi_0)(\yv)-4\pi\left(\beta_0|q(t)|^{2\sigma}+\frac{\gamma-\log 2}{2\pi}\right)q(t) \\ 
		                                                                          +\frac{i\pi q(t)}{2}+(\gamma + \log t) q(0).
		}
		Consequently, using \eqref{eq:forzcont1}-\eqref{eq:forzcont3} and the fact that $\psi_0\in\dom$, we obtain
		\[
		 4\pi(U_0(t)\psi_0)(\yv)=q(0)e^{it}(-\gamma-\log t)+D_{\psi_0}(t)
		\]
		with $D_{\psi_0}\in \clog[0,T] \cap H^{1/2}(0,T)$ for every $\beta\in\R^+$. On the other hand, from Lemma \ref{lemma:qlogcont}, $q\in \clog[0,T] \cap H^{1/2}(0,T)$, for every $\beta\in(0,1]$, as well as $|q|^{2\sigma}q$ (arguing as in the first part of the proof of Lemma \ref{lem-lip}). Hence, observing also that the function $(e^{it}-1)(-\gamma-\log t)\in H^1(0,T)$ and combining all these facts with \eqref{eq-aux}, we obtain that $B_{q}\in \clog[0,T] \cap H^{1/2}(0,T)$ for every $\beta\in(0,1]$ and every $0< T < T_* $.}
	\end{proof}

	\begin{proof}[Proof of Theorem \ref{teo:conservation}] 
		The proof is divided into two parts, where we prove mass and energy conservation separately.
	
		\medskip
		
		{\it Part 1.} Let us first consider the mass conservation. Using the Fourier transform, \eqref{eq:ansatz} reads
		\[
  			\widehat{\psi_t}(\pv)=e^{-ip^2t}\widehat{\psi_0}(\pv)+\frac{ie^{-i\pv\cdot\yv}}{2\pi}\int_0^t\dtau\:e^{-ip^2(t-\tau)}q(\tau).
 		\]
 		Hence,
 		\bmln{
  			\lf|\widehat{\psi_t}(\pv) \ri|^2 = \lf| \widehat{\psi_0}(\pv) \ri|^2+\frac{1}{\pi}\mathrm{Im}\bigg\{e^{i\pv\cdot\yv}\widehat{\psi_0}(\pv)\int_0^t\dtau\:e^{-ip^2\tau}q^*(\tau)\bigg\}\\
  			+\frac{1}{4\pi^2}\int_0^t\dtau\int_0^t\ds\:e^{-ip^2(s-\tau)}q(\tau)q^*(s),
 		}
 		so that, denoting by $\TF^{-1}$ the anti-Fourier transform on $\R^2$,
 		{\bmln{
  		\lf\|\psi_t \ri\|_{L^2(\R^2)}^2 = \lf\|\psi_0 \ri\|_{L^2(\R^2)}^2 + 2\mathrm{Im}\left\{\TF^{-1}\bigg\{ \widehat{\psi_0}(\pv)\int_0^t\dtau\:e^{-ip^2\tau}q^*(\tau)\bigg\}(\yv)\right\}	\\ 
                            	 +\frac{1}{2\pi}\TF^{-1}\left\{\int_0^t\dtau\int_0^t\ds\:e^{-ip^2(s-\tau)}q(\tau)q^*(s)\right\}(\zev) 
                           =:\lf\|\psi_0\ri\|_{L^2(\R^2)}^2+\Psi+\Phi.
		}}
 
 		Furthermore, by the properties of the Fourier transform and the definition of $U_0$,
 		\bdm
  			\Psi =  2\mathrm{Im}\left\{\int_0^t\dtau\:q^*(\tau)\TF^{-1}\big\{e^{-ip^2\tau}\widehat{\psi_0}(\pv)\big\}(\yv)\right\} = 2\mathrm{Im}\left\{\int_0^t\dtau\:q^*(\tau)(U_0(\tau)\psi_0)(\yv)\right\},
		\edm
 		so that by \eqref{eq:charge_inv} proven in Lemma \ref{lem-inv},
 		\[
  			\Psi = \Psi_1+\Psi_2	 := -\frac{1}{4}\int_0^t\dtau\:|q(\tau)|^2+\frac{1}{2\pi}\mathrm{Im}\bigg\{\int_0^t\dtau\:q^*(\tau)\frac{\diff}{\diff \tau}\int_0^\tau\ds\:(-\gamma-\log(\tau-s)q(s))\bigg\}.
		\]
 		Now, we can prove that $\Psi+\Phi=0$. First, one sees that
 		\[
  			\Phi=\frac{1}{\pi}\TF^{-1}\left\{\mathrm{Re}\bigg\{\int_0^t\dtau\:q^*(\tau)\int_0^\tau\ds\:q(s)e^{-ip^2(\tau-s)}\bigg\}\right\}(\zev),
 		\]
 		then we compute
 		\[
  			\int_0^\tau\ds\:q(s)e^{-ip^2(\tau-s)}=i\frac{d}{d\tau}\int_0^\tau\ds\:\frac{e^{-ip^2(\tau-s)}}{p^2+1}q(s)-\frac{iq(\tau)}{p^2+1}+\int_0^\tau\ds\:\frac{e^{-ip^2(\tau-s)}}{p^2+1}q(s),
 		\]
		thus obtaining
		\bmln{
  			\Phi = 	-\frac{1}{\pi}\TF^{-1}\left\{\mathrm{Im}\bigg\{\int_0^t\dtau\:q^*(\tau)\frac{\diff}{\diff \tau}\int_0^\tau\ds\:\frac{e^{-ip^2(\tau-s)}}{p^2+1}q(s)\bigg\}\right\}(\zev)	\\
           	+\frac{1}{\pi} \TF^{-1} \left\{\mathrm{Re}\bigg\{\int_0^t\dtau\:q^*(\tau)\int_0^\tau\ds\:\frac{e^{-ip^2(\tau-s)}}{p^2+1}q(s) \bigg\} \right\}(\zev).
 		}
 		Hence, using again the properties of the Fourier transform, the above expression can be rewritten as
 		\bmln{
  			\Phi = -\frac{1}{\pi}\mathrm{Im}\left\{\int_0^t\dtau\:q^*(\tau)e^{i\tau}\frac{\diff}{\diff \tau}\bigg\{ e^{-i\tau}\int_0^\tau\ds\:q(s)\TF^{-1}\bigg[ \frac{e^{-ip^2(\tau-s)}}{p^2+1} \bigg] \bigg\}(\zev)\right\}	\\ 
  			= -\frac{1}{2\pi}\mathrm{Im}\left\{\int_0^t\dtau\:q^*(\tau)e^{i\tau}\frac{\diff}{\diff \tau}\int_0^\tau\ds\:q(s)e^{-is}(-\gamma-\log(\tau-s))\right\} \\	
          		-\frac{1}{2\pi^2}\mathrm{Im}\left\{\int_0^t\dtau\:q^*(\tau)e^{i\tau}\frac{\diff}{\diff \tau}\int_0^\tau\ds\:q(s)e^{-is}Q(1;\tau-s)\right\}  =:  \Phi_1+\Phi_2,
		}
 		where we have made use of \eqref{eq:sici} and \eqref{eq-Q}. {Notice that the quantity} 
 		\bdm
 			{e^{i \tau} \int_0^\tau\ds\:q(s)e^{-is}(-\gamma-\log(\tau-s))}
		\edm
		{is a priori only continuous and therefore its derivative must be interpreted in distributional sense via}		
		\bml{
			\label{eq: stanerchia}
  			{e^{i\tau}\frac{\diff}{\diff \tau}\int_0^\tau\ds\:q(s)e^{-is}(-\gamma-\log(\tau-s))=   \frac{\diff}{\diff \tau}\int_0^\tau\ds\:q(s)(-\gamma-\log(\tau-s))}	\\
                                                                               	{-\int_0^\tau\ds\:q(s)\frac{e^{i(\tau-s)}-1}{\tau-s}}
 		}
 		{whose regularity is proven in Lemma \ref{lemma:dotq}.}
Now, with some computations, one finds that
 		\bmln{
  			\Phi_2+ \Psi_2 =-\frac{1}{2\pi^2}\mathrm{Im}\bigg\{\int_0^t\dtau\:q^*(\tau)\int_0^\tau\ds\:q(s)e^{i(\tau-s)}\dot{Q}(1;\tau-s)\bigg\}	\\
  			= -\frac{1}{2\pi}\mathrm{Im}\bigg\{\int_0^t\dtau\:q^*(\tau)\int_0^\tau\ds\:q(s)\frac{e^{i(\tau-s)}-1}{\tau-s}\bigg\},
 		}
 		since $ \dot{Q}(1;\tau-s)={\pi} \frac{1 - e^{-i(\tau-s)}}{\tau-s} $, as it follows from \eqref{eq:sici} and the definition of the sine and cosine integral functions \cite[Eqs. 5.2.1 \& 5.2.2]{AS}. Similarly, \eqref{eq: stanerchia} leads to $\Phi_2 + \Psi_1 + \Phi_1 = - \Psi_2 $ and thus completing the proof of the mass conservation.				
		
		\medskip
		
		{\it Part 2.} Let us turn our attention to energy conservation. Since $\psi_0\in\dom[\F]$, taking $\lambda=1$, \eqref{eq:ansatz} yields	
		\[
 			\psi_t(\xv) =  \lf(U_0(t)\phi_{1,0}\ri)(\xv)+\frac{q(t)}{2\pi}K_0(|\xv-\yv|)-\frac{1}{2\pi}\int_0^t\dtau\: \lf(\dot{q}(\tau)-iq(\tau) \ri) \lf( U_0(t-\tau) K_0\lf(|\cdot -\yv|\ri)\ri)(\xv),
 		\]
 		where we have integrated by parts and used the simple formula
 		\[
 			\frac{\diff}{\diff \tau} \lf[e^{-i(t - \tau)}  \lf( U_0(t - \tau) K_0 \ri)(\xv) \ri] = i e^{- i (t - \tau)}U_0(t - \tau; \xv),
		\]
		which can be easily verified by rewriting the quantities via Fourier transform. In light of Remark \ref{rem:meaning2}, {Lemma \ref{lemma:forzcont} and Lemma \ref{lemma:dotq}}, the terms involving $ \dot q $ have to be understood as discussed in Remark \ref{rem:meaning2}, i.e., as the pairing between a function in $ H^{-1/2}(\R) $ and another in $ H^{1/2}(\R) $. In the very same way we get
 		\[
  			\phi_{1,t}(\xv)= \lf(U_0(t)\phi_{1,0}\ri)(\xv)-\frac{1}{2\pi}\int_0^t\dtau\: \lf(\dot{q}(\tau)-iq(\tau) \ri) \lf( U_0(t-\tau) K_0\lf(|\cdot -\yv|\ri)\ri)(\xv).
 		\]
 		Then, we can compute the $H^1$ norm of $\phi_{1,t}$ as
 		\beq
 		        \label{eq-norm_relation}
  			\lf\|\phi_{1,t} \ri\|_{H^1(\R^2)}^2 =\|\phi_{1,0}\|_{H^1(\R^2)}^2+\widetilde{\Psi}_t+\widetilde{\Phi}_t,
 		\eeq
 		where
 		\beq
 			\label{eq:psit}
  			\widetilde{\Psi}_t =-2\mathrm{Re}\bigg(\int_0^t\dtau\: \lf(\dot{q}(\tau)-iq(\tau) \ri)^*(U_0(\tau)\phi_{1,0})(\yv)\bigg)
 		\eeq
 		and 
 		{\bml{
 		 \label{eq:phit}
 		 \widetilde{\Phi}_t = \frac{1}{2\pi}\int_0^t\dtau\:(\dot{q}(\tau)-iq(\tau))^*\bigg(\int_0^t\ds\:(U_0(\tau-s)K_0(|\cdot\,-\yv|))(\yv)(\dot{q}(s)-iq(s))\bigg)	\\	
 		        	=		\frac{1}{\pi}\Re\bigg\{\int_0^t\dtau\:(\dot{q}(\tau)-iq(\tau))^*\bigg(\int_0^\tau\ds\:(U_0(\tau-s)K_0(|\cdot\,-\yv|))(\yv)(\dot{q}(s)-iq(s))\bigg)\bigg\}.
		}}
 		Lemmas \ref{lemma:forzcont} {and \ref{lemma:dotq}} guarantee that the r.h.s. of \eqref{eq:psit} and \eqref{eq:phit} are well defined expressions, which should be understood as explained in Remark \ref{rem:meaning2}{: indeed, \eqref{eq:sici} yields an explicit expression of the quantity $ (U_0(\tau-s)K_0(|\cdot\,-\yv|))(\yv) $, which is continuous up to a logarithmic term (see also \eqref{eq:a2} and the lines below).} On the other hand, using \eqref{eq:sici}, we can immediately rewrite \eqref{eq:phit} as	(recall the definition of $Q(\cdot\,;\cdot)$ in \eqref{eq-Q})
 		\begin{equation}
 		 \label{eq:cancel}
 		 \widetilde{\Psi}_t+\widetilde{\Phi}_t=\Re\bigg\{\int_0^t\dtau\:(\dot{q}(\tau)-iq(\tau))^*\big[\Psi(\tau)+\Phi(\tau)\big]\bigg\}.
 		\end{equation}
 		where
 		\[
  			\Phi(\tau):= \frac{1}{2\pi^2}\int_0^\tau\ds\:e^{i(\tau-s)}\big[Q(1;\tau-s)-\pi(\gamma+\log(\tau-s))\big](\dot{q}(s)-iq(s))
		\]
		and $\Psi(\tau)=-2(U_0(\tau)\phi_{1,0})(\yv)$,
 		so that we have to better investigate the quantity $\Psi(\tau)+\Phi(\tau)$. Now, using again $\psi_0\in\dom[\F]$ and \eqref{eq:charge_inv}, we find that
 		\bmln{
 			\lf( U_0(\tau)\phi_{1,0} \ri)(\yv) =  -\frac{q(0)}{2\pi} \lf(U_0(\tau)K_0(|\cdot-\yv|) \ri)(\yv)+\left(\beta_0|q(\tau)|^{2\sigma}+\frac{\gamma-\log 2}{2\pi}\right)q(\tau) -\frac{iq(\tau)}{8}	\\ 
 			 +\frac{q(0)}{4\pi}(-\gamma-\log \tau)+\frac{1}{4\pi}\int_0^\tau\ds\:(-\gamma-\log(\tau-s))\dot{q}(s)
 		}
 		where Lemma \ref{lemma:dotq} guarantees the well-posedness of last term. Plugging into the definition of $\Psi(\tau)$, we can split it into four terms as $ \Psi(\tau) =\Psi_1(\tau)+\Psi_2(\tau)+\Psi_3(\tau)+\Psi_4(\tau)$, where 
 		\[
 		 \begin{array}{l}
  		  \displaystyle \Psi_1(\tau):= \frac{q(0)}{\pi}(U_0(\tau)K_0(|\cdot\,-\yv|))(\yv),\\
  		  \displaystyle \Psi_2(\tau):=-2\left(\beta_0|q(\tau)|^{2\sigma}+\frac{\gamma-\log2}{2\pi}\right)q(\tau),\\
  		  \displaystyle \Psi_3(\tau):=-\frac{q(0)}{2\pi}(-\gamma-\log\tau)+\frac{iq(\tau)}{4}=:\Psi_{3,r}(\tau)+\Psi_{3,i}(\tau),\\
  		  \displaystyle \Psi_4(\tau):=-\frac{1}{2\pi}\int_0^\tau\ds\:(-\gamma-\log(\tau-s))\dot{q}(s)
  		 \end{array}
 		\]
		(arguing as before one can see that all the previous expressions are well-posed{, via, e.g., Lemma \ref{lemma:dotq}}).

		Furthermore, using the definition of the free propagator $U_0(\tau)$ and \eqref{eq:sici}, one sees that
 		\[
  		 \Psi_1(\tau)=\frac{q(0)e^{i\tau}}{2\pi}\big(-\gamma-\log\tau+\tfrac{1}{\pi}Q(1;\tau)\big),
 		\]
 		with $Q $ defined by \eqref{eq-Q}, and thus (summing and subtracting $(-\gamma-\log\tau)$ and setting $\ell(t):=(e^{it}-1)(-\gamma-\log t)$)
 		\[
 		 \Psi_1(\tau)+\Psi_{3,r}(\tau)=\frac{q(0)}{2\pi}\big(\ell(\tau)+\tfrac{e^{i\tau}}{\pi}Q(1;\tau)\big)=:R_1(\tau).
 		\]
 		On the other hand, we observe that $ \Phi(\tau) + \Psi_4(\tau) = R_2(\tau) + R_3(\tau) $, where
 		\begin{align*}
  			R_2(\tau) & := \frac{1}{2\pi}\int_0^\tau\ds\:\dot{q}(s)\left[\ell(\tau-s)+\tfrac{e^{i(\tau-s)}}{\pi}Q(1;\tau-s)\right],\\
  			R_3(\tau) & := -\frac{i}{2\pi}\int_0^\tau\ds\:q(s)e^{i(\tau-s)}\lf[-\gamma-\log(\tau-s)+\tfrac{1}{\pi}Q(1;\tau-s) \ri].
 		\end{align*}
 		As a consequence, we have that $ \Psi_1(\tau) + \Psi_3(\tau) + \Psi_4(\tau) +  \Phi(\tau)  = R_1(\tau)  + R_2(\tau) + R_3(\tau) + \Psi_4(\tau) =:\Gamma(\tau)$. Now, an integration by parts shows that
 		\bmln{
  			\int_0^\tau\ds\:\dot{q}(s) \left[ \ell(\tau-s)+\tfrac{e^{i(\tau-s)}}{\pi}Q(1;\tau-s)\right]= -\tfrac{i\pi}{2} q(\tau) - q(0)\left(\ell(\tau)+\tfrac{e^{i\tau}}{\pi}Q(1;\tau)\right)	\\                                                    
			+\tfrac{i}{\pi}\int_0^\tau\ds\:q(s)e^{i(\tau-s)}Q(1;\tau-s)+ \int_0^\tau\ds\:q(s)\left[\dot{\ell}(\tau-s)+\tfrac{e^{i(\tau-s)}}{\pi}\dot{Q}(1;\tau-s)\right]
 		}
 		and then, plugging into the definition of $R_2(\tau)$, there results
 		\[
  			\Gamma(\tau) = -\frac{i}{2\pi}\int_0^\tau\ds\:q(s)e^{i(\tau-s)}(-\gamma-\log(\tau-s))+\frac{1}{2\pi}\int_0^\tau\ds\:q(s) \left[ \dot{\ell}(\tau-s)+\tfrac{e^{i(\tau-s)}}{\pi}\dot{Q}(1;\tau-s)\right].
 		\]
 		However, easy computations (see \eqref{eq:sici}) exploiting the definition of the trigonometric integral functions (see, e.g., \cite{AS,GR}) yield
 		\[
  			\dot{\ell}(t)+\frac{e^{i t}}{\pi}\dot{Q}(1;t)=i e^{it}(-\gamma-\log t)
 		\]
 		and thus $\Gamma\equiv0$.
 		
 		Summing up, we proved that
 		\begin{equation}
 		 \label{eq-simplif}
 		 \Psi(\tau)+\Phi(\tau)=\Psi_{{2}}(\tau)
 		\end{equation}
 		In addition, observing that
 		\[
 			2\mathrm{Re} \left[(\dot{q}(\tau)-iq(\tau))^*|q(\tau)|^{2\sigma}q(\tau)\right]=\frac{1}{\sigma+1}\frac{\diff}{\diff \tau}\lf( |q(\tau)|^{2\sigma+2} \ri)
 		\]
 		and
 		\[
 			2\mathrm{Re}\left[(\dot{q}(\tau)-iq(\tau))^*q(\tau)\right]=\frac{\diff}{\diff \tau}|q(\tau)|^{2},
 		\]
 		one can see that
 		\bmln{
 		 	\Re\bigg\{\int_0^t\dtau\:(\dot{q}(\tau)-iq(\tau))^*\Psi_{{2}}(\tau)\bigg\} =  -\bigg(\frac{\beta_0|q(t)|^{2\sigma}}{\sigma+1}+\frac{\gamma-\log 2}{2\pi}\bigg)|q(t)|^2 \\ 
 		                                                                                            + \bigg(\frac{\beta_0|q(0)|^{2\sigma}}{\sigma+1}+\frac{\gamma-\log 2}{2\pi}\bigg)|q(0)|^2.
 		}

 		Consequently, combining the last identity with by \eqref{eq-norm_relation}, \eqref{eq:cancel} and \eqref{eq-simplif},we get
 		\[
  			\lf\|\phi_{\lambda,t} \ri\|_{H^1(\R^2)}^2 = \lf\|\phi_{\lambda,t} \ri\|_{H^1(\R^2)}^2-\lf( \tfrac{\beta_0}{\sigma+1}|q(t)|^{2\sigma}+\tfrac{\gamma-\log 2}{2\pi} \ri) |q(t)|^2+ \lf(\tfrac{\beta_0}{\sigma+1}|q(0)|^{2\sigma} + \tfrac{\gamma-\log 2}{2\pi} \ri)|q(0)|^2
 		\]
 		Finally, in view of \eqref{eq:energy}, this means that $E(t)=E(0)$, for any $ t \leq T < T_* $.
	\end{proof}
	

\subsection{Global well-posedness and blow-up alternative}

	Finally, a simple combination of the energy conservation and the blow-up alternative principle allows to prove Theorem \ref{teo:global} and Proposition \ref{pro:blow-up}.

	\begin{proof}[Proof of Theorem \ref{teo:global}] 
		 {Preliminarily, we note that, in view of Proposition \ref{pro:extension q},} the energy conservation proven in Theorem \ref{teo:conservation} holds true up to any $ T  < T_{*} $ {(with $ T_* $ provided by Proposition \ref{pro:charge continuous})}. Moreover, it yields that, if $\beta_0>0$
 		\[
 			\lf| q(t) \ri|\leq C < + \infty, \qquad	\forall t\in[0,T],
		\]
		and any $T < T_{*} ${, since the function $ g(x) = \frac{\beta_0}{\sigma + 1} x^{2\sigma+2} + \frac{\gamma-\log 2}{2\pi} x^2 $ diverges as $ x \to + \infty $.}
		
		Hence, since $ q $ remains bounded as $ t \to T $ by a quantity which is independent of  $T $, it must be (recall that $ T_{*} $ is by definition the maximal existence time of $ q(t) $)
 		\[
  			\limsup_{t\to T^*}|q(t)| \leq C < +\infty,
 		\]
 		which implies that $q$ can be extended to the whole positive half-line and that $q$ is the unique solution of \eqref{eq:charge_eq} in $ C[0,\infty) $, i.e., it is global in time (see \cite[Theorem 2.3]{M}). In addition, Proposition \ref{pro:extension q}  implies that $q\in H^{1/2}(0,T)$, for every finite $T>0$. 
 
 		Consequently, arguing as before, one can prove that the function $\psi_t$ defined by \eqref{eq:ansatz} and \eqref{eq:charge_eq} is in $\dom[\F]$ and solves \eqref{eq:cauchyweak} for every $t\geq0$, thus proving Theorem \ref{teo:global}.
	\end{proof}
	
	\begin{proof}[Proof of Proposition \ref{pro:blow-up}]
 		If $\beta_0<0$, then we have the following alternative: either $ \limsup_{t \to T_*} |q(t)| < +\infty $, which implies that $ T_* = + \infty$ and the solution is global in time, or
 		\bdm
 			\limsup_{t \to T_*} |q(t)| = +\infty.
		\edm
		In this second case we can still have two opposite alternatives: either $ T_* = +\infty $ and, in spite of not being bounded, the solution is nevertheless global in time, or $ T_* < + \infty $ and the blow-up occurs. Indeed, by the energy conservation and the diverging limit of $ q $, we obtain
 		\[
  			\limsup_{t\to T_*} \lf\|\phi_{\lambda,t} \ri\|_{H^1(\R^2)}= +\infty,
 		\]
 		i.e., $\psi_t$ blows-up at a finite time.
	\end{proof}



\end{document}